\documentclass[onecolumn, draftcls]{IEEEtran}
\usepackage{amsmath,amssymb}
\usepackage{graphicx}
\usepackage{mathrsfs}
\usepackage{lineno}
\usepackage{hyperref}
\usepackage{xcolor}
\usepackage{listings}
\usepackage[linesnumbered,ruled,vlined]{algorithm2e}
\usepackage{setspace}
\usepackage{threeparttable}
\usepackage{booktabs}
\usepackage[table]{xcolor}
\usepackage{multirow}
\usepackage{longtable}
\usepackage{diagbox} 
\usepackage{tikz}
\usetikzlibrary{trees,arrows}
\usepackage{makecell}
\usepackage{float}
\usepackage{textgreek}
\usepackage{enumitem}
\usepackage{bm}

\newcommand{\biggg}{\bBigg@{3}}

\makeatother
\SetCommentSty{myCommentStyle}	

\SetKwProg{Fn}{Function}{:}{end}
\SetKwFunction{Choose}{Choose}

\newtheorem{theorem}{Theorem}[section]
\newtheorem{lemma}{Lemma}[section]

\newtheorem{proposition}{Proposition}[section]

\newtheorem{define}{Definition}[section]
\newtheorem{proc}{Procedure}[section]
\newtheorem{example}{Example}[section]
\newtheorem{innercustomgeneric}{Problem}
\newenvironment{problem}[1]
  {\renewcommand\theinnercustomgeneric{#1}\innercustomgeneric}
  {\endinnercustomgeneric}

\newtheorem{remark}{Remark}[section]

\def\zero{{\mathrm{Zero}}}
\def\proj{{\mathrm{Proj}}}
\def\trunc{{\mathrm{Trunc}}}
\def\P{{\mathcal{P}}}

\def\T{{\mathcal{T}}}
\def\Zero{{\mathrm{Zero}}}
\def\x{{\mathbf{x}}}
\def\can{{\mathrm{can}}}

\def\subs{{\mathrm{subs}}}

\lstdefinelanguage{Maple}{
   keywords={if, while, do, else, end, for, from, to,then},
   keywordstyle=\color{blue}\bfseries,
   ndkeywords={class, export, boolean, throw, implements, import, this},
   ndkeywordstyle=\color{darkgray}\bfseries,
   identifierstyle=\color{black},
   sensitive=false,
   comment=[l]{//},
   morecomment=[s]{/*}{*/},
   commentstyle=\color{purple}\ttfamily,
   stringstyle=\color{red}\ttfamily,
   morestring=[b]',
   morestring=[b]"
}

\lstdefinelanguage{SOStools}{
   keywords={syms,sosprogram,monomials,sosineq,sossetobj,sossolve,sosgetsol,sospolyvar},
   keywordstyle=\color{blue}\bfseries,
   ndkeywords={syms,sosprogram,monomials,sosineq,sossetobj,sossolve,sosgetsol},
   ndkeywordstyle=\color{blue}\bfseries,
   identifierstyle=\color{black},
   sensitive=false,
   comment=[l]{//},
   morecomment=[s]{/*}{*/},
   commentstyle=\color{purple}\ttfamily,
   stringstyle=\color{red}\ttfamily,
   morestring=[b]',
   morestring=[b]"
}



\usepackage{cite}

\begin{document}
\title{An Algebraic Method for Full-Rank Characterization in Binary Linear Coding}

\author{Mingyang Zhu, Laigang Guo, Zhenyu Huang, Xingbing Chen, Jue Wang, Tao Guo, and Xiao-Shan Gao
\thanks{M. Zhu is with the Institute of Network Coding, The Chinese University of Hong Kong, Hong Kong, SAR, China (e-mail: mingyangzhu@cuhk.edu.hk).}
\thanks{L. Guo, X. Chen, and J. Wang are with the School of Mathematical Sciences, Beijing Normal University, Beijing, China (e-mail: lgguo@bnu.edu.cn; xb.chen@mail.bnu.edu.cn; w.jue@mail.bnu.edu.cn). \emph{(Corresponding author: Laigang Guo.)}}
\thanks{Z. Huang is with the State Key Laboratory of Cyberspace Security Defense, Institute of Information Engineering, Chinese Academy of Sciences, Beijing, China, and also with the School of Cyber Security, University of Chinese Academy of Sciences, Beijing, China (e-mail: huangzhenyu@iie.ac.cn).}
\thanks{T. Guo is with the School of Cyber Science and Engineering, Southeast University, Nanjing, China (e-mail: taoguo@seu.edu.cn).}
\thanks{X. Gao is with the KLMM, AMSS, Chinese Academy of Sciences, Beijing, China (e-mail: xgao@mmrc.iss.ac.cn).}
}

\maketitle











\begin{abstract}
In this paper, we develop a characteristic set (CS)-based method for deriving full-rank equivalence conditions of symbolic matrices over the binary field. Such full-rank conditions are of fundamental importance for many linear coding problems in communication and information theory. Building on the developed CS-based method, we present an algorithm called Binary Characteristic Set for Full Rank~(BCSFR), which efficiently derives the full-rank equivalence conditions as the zeros of a series of characteristic sets. In other words, the BCSFR algorithm can characterize all feasible linear coding schemes for certain linear coding problems (e.g., linear network coding and distributed storage coding), where full-rank constraints are imposed on several symbolic matrices to guarantee decodability or other properties of the codes. The derived equivalence conditions can be used to simplify the optimization of coding schemes, since the intractable full-rank constraints in the optimization problem are explicitly characterized by simple triangular-form equality constraints.
\end{abstract}


\section{Introduction}
Linear coding has served as a cornerstone of modern information infrastructure, ensuring data reliability and efficiency across diverse domains such as communication systems, distributed storage systems, and resilient cloud computing clusters~\cite{lin2001error,yeung2008information,modern_coding_theory,dimakis2010network,locality2012,lrc2014,speeding2018,yu2019lagrange}. In general, a linear code can be defined by a collection of linear maps. For our analysis, we represent these maps concretely as a collection of matrices, ${\mathscr C} = ({\bf G}_1,\ldots,{\bf G}_N)$, over a finite field $\mathbb{F}$. The elements in ${\bf G}_i$, $1\le i\le N$, are referred to as \emph{encoding coefficients}. For example, a channel code for point-to-point communication can be defined by ${\mathscr C} = ({\bf G}_1)$, where ${\bf G}_1$ is the generator matrix of the channel code~\cite{lin2001error}; a linear network code~(LNC) for network communication can be defined by ${\mathscr C} = ({\bf G}_1,\ldots,{\bf G}_N)$, where ${\bf G}_i, 1\le i\le N$, is the local encoding matrix for the node~$i$ in the network~\cite{yeung2008information}.

In many coding problems, it is of fundamental interest to find all feasible ${\bf G}_i$, $1\le i \le N$, under which a given objective is optimal or near-optimal. Clearly, this important problem can be formalized as a mathematical optimization problem, which is rendered exceptionally difficult by its combinatorial nature. Since a practical linear code should satisfy some \emph{basic constraints}, such as constraints on decodability, matrix structures, information-theoretic security, etc., the optimization problem is likely to be more tractable if we can characterize all \emph{feasible} ${\mathscr C}$ which satisfy these basic constraints.
The set of all feasible ${\mathscr C}$ is referred to as the \emph{feasible set of the linear coding problem}. 
To efficiently solve the optimization problems mentioned above, explicit constraints (such as the zeros of a polynomial system) are necessary.
However, many basic constraints in a linear coding problem are intractable, i.e., they are difficult to be expressed in a \emph{good} form that can be easily utilized by the optimizer. 
For example, the decodability constraint may require certain symbolic matrices to have full rank, while the information-theoretic security constraint may require certain mutual information quantities to be zero. Although the full-rank constraint can be described by requiring the maximal minors to be nonzero, and the security constraint can be described in terms of relations among linear spaces, these descriptions cannot be effectively utilized by the optimizer.
It is non-trivial to explicitly characterize the feasible set of a linear coding problem in a \emph{good} form, e.g., characterizing the feasible set by zeros of simple triangular-form equations, which is precisely the main work of this paper.

\emph{Full-rank} constraints on some specific matrices determined by the encoding coefficients are a class of basic constraints for linear coding. For example, full-rank constraints are involved in the optimal decoding of LNCs~\cite{li2003linear,koetter2003alg} and erasure codes~\cite{modern_coding_theory}, the minimum Hamming distance of linear codes~\cite{lin2001error}, and many other linear coding problems. In coding theory, encoding coefficients are often uniquely determined by \emph{constructive methods}. For example: 
\begin{enumerate}
    \item LNCs: Several algorithms, such as~\cite{li2003linear,jaggi2005polynomial,langberg2006complexity}, are capable of producing a feasible code $\mathscr{C}$ that satisfies the full-rank constraints, which guarantees successful data transmission in a communication network.
    \item Algebraic codes: Many algebraic approaches can construct linear codes with explicit minimum Hamming distance. The constraint of minimum Hamming distance is equivalent to certain full-rank constraints on the generator matrix or parity-check matrix~\cite{lin2001error}.
\end{enumerate}
However, these constructive methods are developed to generate one set of feasible encoding coefficients, rather than giving all sets of feasible encoding coefficients. In other words, they are misaligned with the goal of characterizing the feasible set of the linear coding problem. Moreover, the application of these constructive methods is limited by additional constraints (possibly due to the limitations of a practical system). For instance, in the context of LNC construction, Jaggi-Sanders algorithm~\cite{jaggi2005polynomial} is inapplicable when the algebraic operations available to network nodes are restricted.


Accordingly, a good characterization of the feasible set is essential for optimization. In particular, the fundamental full-rank constraints need to be transformed into more tractable algebraic representations. In this paper, we transform the problem of characterizing the full-rank constraints arising from the linear coding problem to the problem of solving a system of polynomial equations. Owing to the structural properties of the transformed problem, the characteristic set (CS) method proves to be an effective tool. Specifically, we develop a CS–based method to explicitly characterize the feasible set of a linear coding problem by the zeros of a series of characteristic sets.


The CS method is a fundamental tool for the study of systems of polynomial equations \cite{Wu1986,Wang1993,LinLiu1993,LAZARD1991147,KALKBRENER1993143,ChouGao1990,Aubryetal1999}. Its central idea is to transform a general polynomial system into a collection of triangular sets, such that the solution set of the original system can be expressed as a union of the zero sets of these triangular systems. This transformation reduces the problem of solving multivariate polynomial systems to that of solving cascaded univariate equations. Beyond equation solving, the CS method has been widely applied to computing algebraic invariants such as dimension, degree, and order, to addressing radical ideal membership problems, and to proving results in elementary and differential geometry.

Most existing studies on CS methods assume that the solutions are taken over an infinite algebraically closed field. Extensions of the CS method to polynomial systems over finite fields have been investigated in \cite{CHAI2008,GAO2012655,HUANG202166}. Furthermore, a parallel version of the CS algorithm proposed in \cite{GAO2012655} was developed and efficiently implemented in a high-performance computing environment in \cite{ZhaoSongetal2018}. Mou et al. \cite{MOU2021108} proved that the chordal graph structure of a polynomial set is preserved throughout the top-down triangular decomposition process, thereby providing a theoretical foundation for exploiting sparsity and for designing more efficient triangular and regular decomposition algorithms. In particular, Huang et al. \cite{HUANG202166} proposed the characteristic set method for solving Boolean polynomial systems over a binary field and analyzed the computational complexity of the CS method. In this paper, the special triangular form of characteristic sets over the binary field plays a crucial role in the simplified characterization of full-rank constraints. 

The specific contributions of this paper are summarized as follows:
\begin{enumerate}
    \item We develop a CS-based method for deriving a full-rank characterization of symbolic matrices over the binary field. This method builds the theoretical connection between a full-rank problem and a zero-decomposition problem.
    \item Based on the developed method, we present an algorithm called Binary Characteristic Set for Full Rank~(BCSFR), which efficiently derives the full-rank equivalence conditions of symbolic matrices as the zeros of a series of characteristic sets. 
    The BCSFR algorithm provides a convenient characterization of feasible points (i.e., feasible choices of $\mathscr{C}$) in linear coding problems, thereby greatly simplifying the associated optimization problems.
    \item We provide several experiments to verify the effectiveness of the BCSFR algorithm. Specifically, we show that the BCSFR algorithm can be used to find optimal linear network coding and distributed storage coding schemes in practically significant scenarios where no effective solutions exist.
\end{enumerate}

\section{The Optimization Problem in Linear Coding
}
Our discussion in this paper is mainly limited to binary fields. Indeed, the binary field is the most common operating regime for linear coding employed in digital systems.

We consider a linear coding optimization problem involving $n$ binary variables, denoted as Problem~\ref{prob:P1}, in which full-rank constraints play a central role. General solvers for this optimization problem are unavailable unless it is cast into a highly specialized form. Full-rank constraints are intractable for optimization, because the feasible set $\mathcal S$ is described implicitly (e.g., by the rank function) or by very complex equations (e.g., by maximal non-zero minors). It is not straightforward to generate a point belonging to the feasible set $\mathcal S$. Consequently, in the absence of structural insights, solving Problem~\ref{prob:P1} necessitates an exhaustive search over $\mathbb{F}_2^n$, incurring an exponential complexity of $\mathcal{O}(2^n)$ merely to verify the feasibility of all points.

To address this computational challenge, we introduce a reformulation, Problem~\ref{prob:P2}, which replaces the intractable full-rank constraints in Problem~\ref{prob:P1} with an equivalent system of simple triangular-form equations. Specifically, we develop a CS-based method to characterize $\mathcal S$ as the union of the zero sets of \emph{characteristic sets}.\footnote{
The definition of the characteristic set will be given in Definition \ref{def:CS} of this paper. Note that the meaning of the characteristic set here is basically the same as that referred to in Wu's method \cite{Wu1986}.
} This structural shift changes the strategy from ``verification'' (exhaustive search) to ``generation'' (CS-based method). Within our framework, the feasible points are not discovered by trial-and-error but are explicitly parametrized by the \emph{free variables} in the characteristic sets. 

We now formally present the optimization problems P1, P1a and P2. Let $\mathbf{x}=(x_1,\ldots,x_n) \in \mathbb{F}_2^n$ be the vector of optimization variables for Problem~\ref{prob:P1}. Let ${\bf A}_i({\bf x})$,~$i = 1,\ldots,m_1$, be an $\alpha_i \times \beta_i$ matrix in $\bf x$, and let $h_j({\bf x})$, $j = 1,\ldots,m_2$, be polynomials in $\bf x$. In a linear coding problem, the variables $x_1,\ldots,x_n$ are assumed to be the underlying variables that construct the encoding coefficients. Primarily, all ${\bf A}_i({\bf x})$'s are required to have full rank to guarantee the desired properties of the code, such as decodability and minimum Hamming distance. In addition, there may exist some extra linear or nonlinear constraints due to specific limitations of a practical system, represented by the polynomial equations $h_j({\bf x}) = 0, 1\le j\le m_2$. Among all feasible codes, we are interested in an optimal one that minimizes a certain objective function $f({\bf x})$. Now we formulate this optimization problem as
\begin{problem}{P1}\label{prob:P1}
    \begin{align*}
&\min_{\mathbf{x}}~~ f(\mathbf{x})\\
\mbox{\rm subject to}~~ & \text{i) (Full-rank constraints)}~{\bf A}_i({\bf x}) \text{ is of full row rank},~i = 1,\ldots,m_1\\
&\text{ii) (Non-rank constraints)}~h_j(\mathbf{x})=0,~j=1,\ldots,m_2.
\end{align*}
\end{problem}

The physical meaning of Problem~\ref{prob:P1} will be further demonstrated in Example~\ref{ExampleII.1}. 
As aforementioned, for the specific Problem~\ref{prob:P1}, two direct challenges are:

\begin{enumerate}
    \item[1)] How to determine whether the symbolic matrices $\mathbf{A}_1,\ldots,\mathbf{A}_{m_1}$ over the binary field are of full rank?

\item[2)] How to characterize the feasible set explicitly, and further narrow down the range of $\mathbf{x}$ values based on constraints i) and ii)?

\end{enumerate}

\begin{example}\label{ExampleII.1}
    Consider the well-known linear network coding model~\cite{yeung2008information}. Let $\mathcal{G} = (\mathcal{V},\mathcal{E})$ be a directed acyclic graph~(DAG), where $\mathcal{V}$ is the set of nodes, and $\mathcal{E}$ is the set of edges. Let $s \in \mathcal{V}$ be a node, and $\mathcal{U}\subset \mathcal{V}\setminus\{s\}$ be a subset of nodes. In the field of network communication, it is of fundamental interest to design coding schemes to transmit a message generated by node $s$, denoted by a row vector $\bf m$ over some finite field, to the nodes in $\mathcal{U}$. Linear network codes~(LNCs) have been proven in~\cite{li2003linear,koetter2003alg} to be optimal for such communication tasks. The basic idea behind LNC is quite simple: All nodes forward linear combinations of the symbols they possess. This way, the message ${\bf y}$ received by node $u \in \mathcal{U}$ is a linear transformation of $\bf m$, which can be written as ${\bf y}_u = {\bf m} {\bf F}_u$. To recover the message ${\bf m}$ at node $u$, ${\bf F}_u$ is required to have full row rank, which corresponds to a full-rank constraint in i) of Problem~\ref{prob:P1}. In addition, some other constraints on the algebraic operations that a node can take, corresponding to the non-rank constraints in ii) of Problem~\ref{prob:P1}, may exist in practical implementations. 

    To implement coding operations, each node requires a design of several linear combination coefficients, referred to as encoding coefficients. Let $n$ be the total number of encoding coefficients required for the entire network. For networks of non-trivial size, an exhaustive search of the vector space $\mathbb{F}_2^n$ to select optimal encoding coefficients is infeasible. Although there are established algorithms in the field of network coding (e.g.,~\cite{li2003linear,jaggi2005polynomial,langberg2006complexity}) that can provide one feasible selection of the encoding coefficients, these algorithms cannot solve Problem~\ref{prob:P1}. In particular, these algorithms provide only one solution. In addition, these algorithms become inapplicable if there exist constraints on the size of the finite field or on the algebraic operations of nodes, such as constraints ii) in Problem~\ref{prob:P1}.
\end{example}

To smoothly introduce Problem~\ref{prob:P2}, we first provide an equivalent formulation of Problem~\ref{prob:P1}, denoted as Problem~\ref{prob:P1a}. This reformulation does not simplify Problem~\ref{prob:P1}, but connects the original problem and our proposed symbolic algorithm, which will be explained in detail in Section~\ref{Sec-Meth}. Let ${\bm \zeta}_i=(\zeta_{i,1},\ldots,\zeta_{i,\alpha_i}) \in \mathbb{F}_2^{\alpha_i}$ be a vector of binary variables, $i = 1,\ldots,m_1$. We can rewrite Problem~\ref{prob:P1} in an equivalent form:
\begin{problem}{P1a}\label{prob:P1a}
    \begin{align*}
&\min_{\mathbf{x}}~~ f(\mathbf{x})\\
\mbox{\rm subject to}~~ & \text{i) (Full-rank constraints)}~
\{{\bm \zeta}_i:{\bm \zeta}_i \mathbf{A}_i({\bf x}) = {\bf 0}\}=\{{\bf 0}\}\ \mbox{for}\ i = 1,\ldots,m_1
\\
&\text{ii) (Non-rank constraints)}~h_j(\mathbf{x})=0,~j=1,\ldots,m_2.
\end{align*}
\end{problem}

In this paper, we propose a CS-based algorithm, referred to as \emph{Binary Characteristic Set for Full Rank~(BCSFR)}, to efficiently characterize the feasible set of Problem~\ref{prob:P1} by a collection of characteristic sets. To formulate Problem \ref{prob:P2}, we first present a high-level description of this algorithm (the details will be discussed in Section \ref{Sec-Meth}). Let
\begin{equation}\nonumber
    {\mathcal{X}}_{1,i} = \{ \mathbf{x} \in \mathbb{F}_2^n: \{{\bm \zeta}_i:{\bm \zeta}_i \mathbf{A}_i({\bf x}) = {\bf 0}\}=\{{\bf 0}\}\},\quad i = 1,\ldots,m_1,
\end{equation}
and
\begin{equation}\nonumber
    {\mathcal{X}}_{2,j} = \{\mathbf{x} \in \mathbb{F}_2^n: h_j(\mathbf{x}) = 0\},\quad j = 1,\ldots,m_2.
\end{equation}
Then the feasible set of Problem~\ref{prob:P1a} is $\mathcal{S} = \left(\bigcap_{i=1}^{m_1} {\mathcal{X}}_{1,i}\right) \cap \left(\bigcap_{j=1}^{m_2} {\mathcal{X}}_{2,j}\right)$. Let $\T_i$ be a characteristic set and $\zero(\T_i)$ be its zero set (both will be strictly defined in Section~\ref{Sec-Meth}). 
With the input ${\bm \zeta}_i \mathbf{A}_i({\bf x})$'s and $h_j(\mathbf{x})$'s,
our BCSFR algorithm can generate a collection of characteristic sets, $\T^* = \{\T_1,\ldots,\T_s\}$, such that 
\begin{equation}\nonumber
    \bigcup_{i=1}^s \zero(\T_i) = \mathcal{S}
\end{equation}
and
\begin{equation}\nonumber
    \zero(\T_i) \cap \zero(\T_j) = \emptyset,\quad \forall\,i\ne j.
\end{equation}

Based on the above discussions, we now present Problem~\ref{prob:P2}:
\begin{problem}{P2}\label{prob:P2}
\begin{align*}
&\min_{\mathbf{x}}~~ f(\mathbf{x})\\
\mbox{subject to}~~ & ~\bigvee_{i=1}^{s} \mathcal{T}_{i}=0.
\end{align*}
\end{problem}

\section{Binary Characteristic Set Method and the Corresponding Algorithms}\label{Sec-Meth}
In this section, we first develop a binary characteristic set~(BCS) method, which provides a new perspective for solving Boolean equations over binary field. 
Then, we present the main theorems of this paper, which lay the theoretical foundation for the BCSFR algorithm (our main algorithm) proposed in Section~\ref{subsec:alg}.
At last, the BCSFR algorithm is elaborated with an example. The establishment of this algorithm (BCSFR) transforms the problem of determining the full rank of a symbolic matrix into finding the solutions of a system of polynomial equations. This allows us to efficiently find a series of equivalent conditions for a symbolic matrix to be full rank without having to perform complex calculations of symbolic eigenvalues or Gaussian elimination. 

\subsection{Binary Characteristic Set Method}
In this subsection, we introduce some preliminaries on the characteristic set method over the binary field, and then give the BCS procedure (Procedure~\ref{proc:BCS}).
Let ${\bf x} = (x_1,\ldots,x_n) \in \mathbb{F}_2^n$.

\begin{define}\label{def:canonical_form}
Let $f$ be a polynomial in $\mathbf{x}$.
The variable $x_c$ is called the \emph{leading variable} of $f$ if $c$ is the highest index of the variables appearing in $f$, which is denoted by $c={\rm cls}(f)$.
Then, due to $\mathbb{F}_2$, $f$ can be rewritten as
\begin{align}\label{csform-1}
f=I\cdot x_{c}+U,
\end{align}
where $x_{c}$ is the leading variable of $f$, and $I$ and $U$ are polynomials in $\{x_1,\ldots,x_n\}\backslash \{x_{c}\}$. 
The form~\eqref{csform-1} is called the \emph{canonical form} of $f$, denoted by ${\rm can}(f)=I\cdot x_{c}+U$, and $I$ is called the \emph{initial} of $f$, denoted by $\mbox{init}(f)=I$.
\end{define}
\begin{define}\label{def-cs-order}
Let ${\rm tdeg}(I)$ be the total degree of $I$, ${\rm term}(I)$ be the number of terms in $I$, and ${\rm term}(U)$ be the number of terms in $U$. 
Then $f_1\succ f_2$ is called {\it the first CS order} ({\it 1st-CSO}) if one of the following cases holds:
\begin{itemize}
\item[1)] ${\rm tdeg}(I_1)<{\rm tdeg}(I_2)$
\item[2)] ${\rm tdeg}(I_1)={\rm tdeg}(I_2),\ {\rm term}(I_1)<{\rm term}(I_2)$
\item[3)] ${\rm tdeg}(I_1)={\rm tdeg}(I_2),\ {\rm term}(I_1)={\rm term}(I_2),\ {\rm term}(U_1)<{\rm term}(U_2)$
\item[4)] ${\rm tdeg}(I_1)={\rm tdeg}(I_2),\ {\rm term}(I_1)={\rm term}(I_2),\ {\rm term}(U_1)={\rm term}(U_2),\ {\rm cls}(f_1)>{\rm cls}(f_2)$.
\end{itemize}
\end{define}

\begin{define}\label{def-cs-order-2}
Let ${\rm tdeg}(I)$ be the total degree of $I$, ${\rm term}(I)$ be the number of terms in $I$, and ${\rm term}(U)$ be the number of terms in $U$. 
Then $f_1\succ f_2$ is called {\it the second CS order} ({\it 2nd-CSO}) if one of the following cases holds:
\begin{itemize}
\item[1)] ${\rm cls}(f_1)>{\rm cls}(f_2)$
\item[2)] ${\rm cls}(f_1)={\rm cls}(f_2)$,\ ${\rm tdeg}(I_1)<{\rm tdeg}(I_2)$
\item[3)] ${\rm cls}(f_1)={\rm cls}(f_2)$,\ ${\rm tdeg}(I_1)={\rm tdeg}(I_2)$,\ ${\rm term}(I_1)<{\rm term}(I_2)$
\item[4)] ${\rm cls}(f_1)={\rm cls}(f_2)$,\ ${\rm tdeg}(I_1)={\rm tdeg}(I_2)$,\ ${\rm term}(I_1)={\rm term}(I_2)$,\ ${\rm term}(U_1)<{\rm term}(U_2)$.
\end{itemize}
\end{define}
By Definitions \ref{def-cs-order} and \ref{def-cs-order-2}, note that 1st-CSO and 2nd-CSO are two partial order relations.
If two polynomials $f_1$ and $f_2$ cannot be compared by 1st-CSO or 2nd-CSO, we say that $f_1$ and $f_2$ are equivalent, denoted as $f_1\sim f_2$.
These two partial order relations will be used in Algorithm~\ref{alg:BCSFR}.

Then, we give two examples for elaborating these two orders.
\begin{itemize}
\item[1)] For $f_1=x_2x_3+x_1+1$ and $f_2=x_1x_2x_4+x_1$, we have
\begin{align*}
&{\rm cls}(f_1)=3, {\rm cls}(f_2)=4,
I_1=x_2, I_2=x_1x_2,
U_1=x_1+1, U_2=x_1,\\
&{\rm tdeg}(I_1)={\rm tdeg}(x_2)=1,
{\rm tdeg}(I_2)={\rm tdeg}(x_1x_2)=2,\\
&{\rm term}(I_1)=1, {\rm term}(I_2)=1,
{\rm term}(U_1)=2, {\rm term}(U_2)=1.
\end{align*}
Based on the {\it 1st-CSO}, we obtain $f_1\succ f_2$, because ${\rm tdeg}(f_1)<{\rm tdeg}(f_2)$.
Based on the {\it 2nd-CSO}, we obtain $f_2\succ f_1$, because ${\rm cls}(f_2)>{\rm cls}(f_1)$.
\item[2)] For $f_1=x_2x_3+x_1+1$ and $f_2=(x_1+1)x_4+x_1$, we have
\begin{align*}
&{\rm cls}(f_1)=3, {\rm cls}(f_2)=4,
I_1=x_2, I_2=x_1+1,
U_1=x_1+1, U_2=x_1,\\
&{\rm tdeg}(I_1)={\rm tdeg}(x_2)=1,
{\rm tdeg}(I_2)={\rm tdeg}(x_1+1)=1,\\&
{\rm term}(I_1)=1, {\rm term}(I_2)=2,
{\rm term}(U_1)=2, {\rm term}(U_2)=1.
\end{align*}
Based on the {\it 1st-CSO}, we obtain $f_1\succ f_2$, because ${\rm tdeg}(f_1)={\rm tdeg}(f_2)$ and ${\rm term}(I_1)<{\rm term}(I_2)$.
Based on the {\it 2nd-CSO}, we obtain $f_2\succ f_1$, because ${\rm cls}(f_2)>{\rm cls}(f_1)$.
\end{itemize}

\begin{define}
Let $\P$ be a set of polynomials in $\bf x$. The zero set of $\P$ is defined as
\begin{equation}\nonumber
        \zero(\P) = \{{\bf x}: {\bf x}\text{ is a zero point for }\P=0\}.
\end{equation}
For a collection of sets of polynomials $\P^* = \{\P_1,\ldots,\P_m\}$, where every $\P_i$ is a set of polynomials in $\bf x$, we use the shorthand notation
\begin{equation}\nonumber
    \zero(\P^*) = \bigcup_{i=1}^m \zero(\P_i).
\end{equation}
\end{define}

\begin{define}
Let $\mathcal{P}$ and $\mathcal{T}$ be two distinct sets of polynomials in $\mathbf{x}$.
Then $\mathcal{T}$ is called a branch of $\mathcal{P}$ if $\Zero(\mathcal{T})\subset \Zero(\mathcal{P})$.
\end{define}

\begin{define}\label{zero-decomposition-alg}
Let ${\mathcal P}$ be a polynomial set in $\x$. 
We say {\it zero-decomposition} to mean that there exist several branches of ${\mathcal P}$, $\T_1,\ldots,\T_s$, such that 
\begin{align}\label{zero-dec}
\zero(\mathcal{P})=\bigcup_{i=1}^{s}\zero(\mathcal{T}_i),
\end{align}
which is denoted by ${\rm zd}(\mathcal{P})=\{\T_1,\ldots,\T_s\}$.
In particular, the decomposition of $\zero(\mathcal{P})$, i.e.~\eqref{zero-dec}, is called a zero orthogonal decomposition if $\zero(\mathcal{T}_i)\cap \zero(\mathcal{T}_j)=\emptyset$ for any $1\le i< j\le s$, and it is denoted by ${\rm zod}(\mathcal{P})=\{\T_1,\ldots,\T_s\}$.
\end{define}
Let $\can(f)=Ix_c+U$.
Then $f=0$ if and only if one of the following branches vanishes:  
\begin{align}\label{ID-1}
\mathcal{T}_1=\{I,U\},\ \mathcal{T}_2=\{I+1,x_c+U\}.
\end{align}

\begin{define}\label{def-inidec}
We say the \emph{initial decomposition} for $f$ to mean that $f$ is divided into two branches $\mathcal{T}_1$ and $\mathcal{T}_2$ in \eqref{ID-1}, which is denoted by ${\rm id}(f)=\{\T_1,\T_2\}$.
\end{define}



\begin{lemma}\label{T1T2}
If $f$ is divided into $\mathcal{T}_1$ and $\mathcal{T}_2$ by initial decomposition, then ${\rm Zero}(f)={\rm Zero}(\mathcal{T}_1)\cup {\rm Zero}(\mathcal{T}_2)$ and ${\rm Zero}({\mathcal T}_1) \cap {\rm Zero}({\mathcal T}_2) = \emptyset$.
\end{lemma}
\begin{IEEEproof}
Choose $\x^* \in \zero(f)$. We have $0 = f(\x^*) = I(\x^*) x_c(\x^*) + U(\x^*)$. If $I(\x^*) = 0$, then $U(\x^*) = 0$, which implies $\x^* \in \zero(\T_1)$. If $I(\x^*) = 1$, then $x_c(\x^*) + U(\x^*) = 0$, which implies $\x^* \in \zero(\T_2)$. Thus, $\x^* \in \zero(\mathcal{T}_1)\cup {\rm Zero}(\mathcal{T}_2)$, which implies $\zero(f) \subset \zero(\mathcal{T}_1)\cup {\rm Zero}(\mathcal{T}_2)$. On the other hand, it is straightforward to see that every $\x^* \in \zero(\T_1)$ or $\x^* \in \zero(\T_2)$ is a zero point of $f$. Thus, $\zero(\mathcal{T}_1)\cup {\rm Zero}(\mathcal{T}_2) \subset \zero(f)$. Combining the two cases, we have ${\rm Zero}(f)={\rm Zero}(\mathcal{T}_1)\cup {\rm Zero}(\mathcal{T}_2)$.

If $\x^* \in {\rm Zero}({\mathcal T}_1) \cap {\rm Zero}({\mathcal T}_2)$, then $I(\x^*) = I(\x^*) + 1 = 0$, which makes a contradiction. Thus, ${\rm Zero}({\mathcal T}_1) \cap {\rm Zero}({\mathcal T}_2) = \emptyset$.
\end{IEEEproof}

\begin{define}\label{def-id4P}
Let $\P=\{f_i,i=1\ldots,m\}$ be a polynomial set in $\x$.
Choose a non-constant $f_i$ from $\P$ and let ${\rm id}(f_i)=\{\T_1,\T_2\}$.
Let $\mathcal{Q}_0=\{I,U\}\cup (\mathcal{P}\backslash \{f_i\})$, $\mathcal{U}=\{x_c+U\}$, and $\mathcal{Q}_1={\rm subs}(x_c=U,\{I+1\}\cup (\mathcal{P}\backslash \{f_i\}))$. We say the {\it initial decomposition} for $\P$ to mean that $\P$ is divided into three branches $\mathcal{Q}_0$, $\mathcal{Q}_1$, and $\mathcal{U}$, denoted by ${\rm id}(\P)=\{\mathcal{Q}_0, \mathcal{Q}_1, \mathcal{U} \}$.
\end{define}

Note that in Definition~\ref{def-id4P}, the order of choosing $f_i$ from $\P$ is not specified. In the algorithms proposed later, we will use the 1st-CSO and the 2nd-CSO defined by Definitions~\ref{def-cs-order} and \ref{def-cs-order-2}, respectively. By Lemma~\ref{T1T2}, we have for ${\rm id}(\P)=\{\mathcal{Q}_0, \mathcal{Q}_1, \mathcal{U} \}$, $\zero(\P) = \zero(\mathcal{Q}_0) \cup \zero (\mathcal{Q}_1 \cup \mathcal{U})$ and $\zero(\mathcal{Q}_0) \cap \zero (\mathcal{Q}_1 \cup \mathcal{U}) = \emptyset$.

\begin{define}\label{def:CS}
Let $\T=\{f_i,i=1,\ldots,m\}$ be a set of polynomials in $\mathbf{x}$, where $\can(f_i)=I_i x_{c_i}+U_i$. Then $\T$ is called a \emph{monic characteristic set} if 
\begin{itemize}
\item[1)] $I_i=1$ 
\item[2)] $c_i<c_j$ for $1\le i<j\le m$
\item[3)] $U_i=g_i(x_{k_1},\ldots,x_{k_{n-m}})$, i.e., $U_i$ is a polynomial in $x_{k_1},\ldots,x_{k_{n-m}}$,
where $$\{k_1,\ldots,k_{n-m}\} = \{1,\ldots,n\}\backslash \{c_1,\ldots,c_m\}.$$
\end{itemize}
Following the above, $|{\T}|=m$ is called the length of $\T$, $\{x_{k_1},\ldots,x_{k_{n-m}}\}$ is called the free variables set of $\T$, denoted by ${\rm fvs}(\T)=\{x_{k_1},\ldots,x_{k_{n-m}}\}$, and $n-m$ is called the degree of freedom (df) of $\T$, denoted by ${\rm df}(\T)=n-m$.
\end{define}

Note that in the binary field, ${\rm df}(\T)$ determines the number of zeros of $\T$, i.e., $|\zero(\T)|=2^{{\rm df}(\T)}$.
For convenience, in the rest of this paper, we say \emph{characteristic set} to mean a monic characteristic set.

\begin{define}\label{def:trunc}
Let ${\bf x} = (x_1,\ldots,x_n)$ and $\tilde{\bf x} = (x_{n+1},\ldots,x_{n+v})$. Let $\mathcal T$ be a characteristic set in $({\bf x}, \tilde{\bf x})$, which has the general form:
\begin{align*}
\{x_{c_i} + U_i: 1\le i\le {m}, 1\le c_1<\cdots<c_{\ell} \le n < c_{\ell+1} < \cdots < c_{m} \le n+v\}.
\end{align*}
Then $\trunc_{\bf x}(\T)$ is called the truncation of $\mathcal{T}$ onto the $\bf x$-variables if
\begin{equation}\nonumber
\trunc_{\bf x}(\T) = \{x_{c_j} + U_j: ~1\le j\le \ell, 1\le c_1<\cdots<c_{\ell}\le n\}.
\end{equation}
\end{define}

\begin{define}
Let $\mathcal{P}$ be a polynomial set in $\x$. Then ${\rm zod}(\mathcal{P})=\{\T_1,\ldots,\T_s\}$ is called a zero orthogonal characteristic decomposition of $\mathcal{P}$ if $\{\T_1,\ldots,\T_s\}$ is a collection of characteristic sets, which is denoted by ${\rm zocd}(\mathcal{P})=\{\T_1,\ldots,\T_s\}$.
\end{define}

Next, we will present a procedure to illustrate the computational process of the BCS method.

\begin{proc}[BCS Procedure]
\label{proc:BCS}
Let ${\mathcal P} = \{f_1,\ldots,f_m\}$ be a polynomial set. 

\textbf{Step 1}. Initialize two collections of polynomial sets $\mathcal{P}^* = \{\P\}$ and $\T^* = \emptyset$.

\textbf{Step 2}. Choose an element $\mathcal{Q}$ from $\mathcal{P}^*$, which is a set of polynomials. Let $\mathcal{P}^* \leftarrow \mathcal{P}^* \setminus \{\mathcal{Q}\}$ and $\T = \emptyset$.

\textbf{Step 3}. Perform initial decomposition for $\mathcal{Q}$. By Definition~\ref{def-id4P}, let ${\rm id}(\mathcal{Q}) = \{\mathcal{Q}_0,\mathcal{Q}_1,\mathcal{U}\}$, where $f = Ix_c + U \in \mathcal{Q}$ is the polynomial chosen for initial decomposition, $\mathcal{Q}_0=\{I,U\}\cup (\mathcal{Q}\backslash \{f\})$, $\mathcal{U}=\{x_c+U\}$, and $\mathcal{Q}_1={\rm subs}(x_c=U,\{I+1\}\cup (\mathcal{Q}\backslash \{f\}))$. We obtain
\begin{align}\label{eq:proc_eq_1}
\zero(\mathcal{Q})=\zero(\mathcal{Q}_0)\cup \zero(\mathcal{Q}_1\cup \mathcal{U})
\end{align}
and
\begin{align}\label{eq:proc_eq_2}
\zero(\mathcal{Q}_0)\cap \zero(\mathcal{Q}_1\cup \mathcal{U}) = \emptyset.
\end{align}

\textbf{Step 4}.
Update $\mathcal{P}^* \leftarrow \mathcal{P}^* \cup \{\mathcal{Q}_0 \cup \T\}$, $\T \leftarrow {\rm subs}(x_c=U,\mathcal{T})\cup\mathcal{U}$, and $\mathcal{Q} \leftarrow \mathcal{Q}_1$. We obtain
\begin{equation}\label{eq:proc_eq_3}
    \zero(\P) = \zero(\P^*) \cup \zero(\T^*) \cup \zero(\mathcal{Q}\cup \T)
\end{equation}
and for every two distinct elements $\mathcal{P}_1$ and $\mathcal{P}_2$ in $\mathcal{P}^* \cup \T^* \cup \{\mathcal{Q}\cup \T\}$
\begin{equation}\label{eq:proc_eq_4}
    \zero(\mathcal{P}_1) \cap \zero(\mathcal{P}_2) = \emptyset.
\end{equation}

\textbf{Step 5}. Repeat Steps 3 and 4 (i.e., perform initial decomposition for $\mathcal{Q}$) until $\mathcal{Q}\setminus \{0\} = \emptyset$. If $\T \ne \emptyset$, let $\T^* \leftarrow \T^* \cup \{\T\}$.

\textbf{Step 6}.
Repeat Steps 2--5 until $\P^* = \emptyset$. We finally obtain a zero orthogonal decomposition of $\mathcal{P}$,
\begin{align}\label{eq:proc_eq_5}
\zero(\mathcal{P})=\bigcup_{i=1}^s \zero(\mathcal{T}_i), 
\end{align}
where each $\mathcal{T}_i$ corresponds to a $\T$ obtained by Step~5, and ${\rm \Zero}({\mathcal T}_i) \cap \Zero({\mathcal T}_j) = \emptyset$ for any $1\le i < j\le s$.
\end{proc}

\begin{remark}
    The relations \eqref{eq:proc_eq_1}--\eqref{eq:proc_eq_5} can be explained as follows: 
    \begin{enumerate}
        \item \eqref{eq:proc_eq_1} and \eqref{eq:proc_eq_2} are consequences of Lemma~\ref{T1T2}.
        \item \eqref{eq:proc_eq_3} and \eqref{eq:proc_eq_4} can be proved by induction based on \eqref{eq:proc_eq_1} and \eqref{eq:proc_eq_2}.
        \item \eqref{eq:proc_eq_5} is a consequence of \eqref{eq:proc_eq_3} and \eqref{eq:proc_eq_4}.
    \end{enumerate}
\end{remark}

Accordingly, we give an example to elaborate on Procedure~\ref{proc:BCS}.
\begin{example}
Consider $\mathcal{P}=\{f_1=x_2x_3+x_1+1,\; f_2=x_1x_2x_4+x_1\}$.
For ease of description, we use the $1st$-CSO in Definition~\ref{def-cs-order}.

\noindent\textbf{Step 1.}
Set $\mathcal{P}^*=\{\mathcal{P}\}$ and $\mathcal{T}^*=\emptyset$.

\noindent\textbf{Round 1 (Step 2--5).}
Step 2: Choose $\mathcal{Q}=\mathcal{P}$ from $\mathcal{P}^*$. Let $\P^* \leftarrow \P^* \setminus \{\mathcal{Q}\} = \emptyset$ and $\mathcal{T}=\emptyset$.

\noindent\textbf{Repeat Step 3--4.}
Step 3: Choosing $f_1 = x_2 x_3 + x_1 + 1$ and performing ${\rm id}(\mathcal Q)$, we obtain
\[
\mathcal{Q}_0=\{x_2, x_1+1, f_2\},\quad
\mathcal{U}=\{x_3+x_1+1\},\quad
\mathcal{Q}_1=\subs(x_3=x_1+1,\{x_2+1, f_2\}).
\]

\noindent{Step 4:} We update
\[
\mathcal{P}^*\leftarrow \mathcal{P}^*\cup\{\mathcal{Q}_0\cup \mathcal{T}\} = \big\{\{x_2, x_1+1, f_2\}\big\},\]
\[\mathcal{T}\leftarrow \subs(x_3=x_1+1,\mathcal{T})\cup\mathcal{U}=\{x_3+x_1+1\},
\]
\[
\mathcal{Q}\leftarrow \mathcal{Q}_1=\{x_2+1,\; x_1x_2x_4+x_1\}.
\]

\noindent{Step 3:} Choosing $x_2+1$ and performing ${\rm id}(\mathcal Q)$, we obtain 
\[
\mathcal{Q}_0=\{1, x_1x_2x_4+x_1\},\quad
\mathcal{U}=\{x_2 + 1\},\quad
\mathcal{Q}_1=\subs(x_2 = 1,\{0,x_1x_4+x_1\}).
\]

\noindent{Step 4:} Since $1 \in \mathcal{Q}_0$ leads to a contradiction, we discard $\mathcal{Q}_0$ and do not need to update $\P^*$. We update
\[
\mathcal{T}\leftarrow \subs(x_2=1,\mathcal{T})\cup\mathcal{U}=\{x_3+x_1+1, x_2+1\},
\]
\[
\mathcal{Q}\leftarrow \mathcal{Q}_1=\{x_1x_4+x_1\}.
\]

\noindent{Step 3:} Choosing $x_1x_4+x_1$ and performing ${\rm id}(\mathcal Q)$, we obtain
\[
\mathcal{Q}_0=\{x_1\},\quad
\mathcal{U}=\{x_4+x_1\},\quad
\mathcal{Q}_1=\subs(x_4=x_1,\{x_1+1\}).
\]

\noindent{Step 4:} Update
\[
\mathcal{P}^*\leftarrow \mathcal{P}^*\cup\{\mathcal{Q}_0\cup \mathcal{T}\}
=\big\{\{x_2, x_1+1, f_2\},\{x_1, x_3+x_1+1, x_2+1\}\big\},
\]
\[
\mathcal{T}\leftarrow \subs(x_4=x_1,\mathcal{T})\cup\mathcal{U}
=\{x_3+x_1+1, x_2+1, x_4+x_1\},\]
\[
\mathcal{Q}\leftarrow \mathcal{Q}_1=\{x_1+1\}.
\]

\noindent{Step 3:}  Choosing $x_1 + 1$ and performing ${\rm id}(\mathcal Q)$, we obtain
\[
\mathcal{Q}_0=\{1\},\quad
\mathcal{U}=\{x_1 + 1\},\quad
\mathcal{Q}_1=\subs(x_1 = 1,\{0\}).
\]

\noindent{Step 4:} Since $1 \in \mathcal{Q}_0$ leads to a contradiction, we discard $\mathcal{Q}_0$ and do not need to update $\P^*$. We update
\[
\mathcal{T}\leftarrow \subs(x_1=1,\mathcal{T})\cup\mathcal{U}=\{x_1+1, x_2+1, x_3, x_4+1\},
\]
\[
\mathcal{Q}\leftarrow \mathcal{Q}_1=\{0\}.
\]

\noindent{Step 5:}
Since $\mathcal{Q} \setminus \{0\} = \emptyset$, append $\mathcal{T}$ to $\mathcal{T}^*$.

\noindent\textbf{Round 2 (Step 2--5).} Step 2: Choose $\mathcal{Q}=\{x_2, x_1+1, f_2 = x_1x_2x_4+x_1\}$ from $\P^*$.
Let $\P^* \leftarrow \P^* \setminus \{\mathcal{Q}\} = \big\{\{x_1, x_3+x_1+1, x_2+1\}\big\}$ and $\mathcal{T}=\emptyset$.

\noindent Step 3--5: Similar to the process of Round 1, we find that this branch $\mathcal{Q}$ has no zeros. Consequently, no new characteristic set will be appended to $\T^*$.


\noindent\textbf{Round 3 (Step 2--5).} Choose $\mathcal{Q}=\{x_1, x_3+x_1+1, x_2+1\}$ from $\P^*$.
Let $\P^* \leftarrow \P^* \setminus \{\mathcal{Q}\} = \emptyset$ and $\mathcal{T}=\emptyset$.

\noindent Step 3--5: Similar to the process of Round 1, we obtain $\P^* = \emptyset$ and $\T = \{x_1, x_2+1, x_3+1\}$ at the end of this round. Append $\T$ to $\T^*$.


\noindent\textbf{Step 6.}
Since $\P^* = \emptyset$, Procedure~\ref{proc:BCS} ends and outputs
\[
\zero(\mathcal{P})=\zero(\{x_1+1,x_2+1,x_3, x_4+1\}) \cup \zero(\{x_1, x_2+1, x_3+1\}).
\]
\end{example}

\subsection{Main Theorems}\label{subsec:main_theorem}
Let ${\bf x} = (x_1,\ldots,x_n)$ and $\tilde{\bf x} = (x_{n+1},\ldots,x_{n+v})$ be two vectors of binary variables, which are referred to as the \emph{${\bf x}$-variables} and the \emph{$\tilde{\bf x}$-variables}, respectively. 
\begin{define}\label{def:x_tilde_linear_poly}
    A polynomial set ${\mathcal P} = \{f_1,\ldots,f_{\ell_1},r_1,\ldots,r_{\ell_2}\}$ in $({\bf x},\tilde{{\bf x}})$ is called an $\tilde{\bf x}$-linear polynomial set if the following conditions are satisfied:
    \begin{enumerate}
    \item $f_1,\ldots,f_{\ell_1}$ are polynomials in $({\bf x},\tilde{\bf x})$
    \item $r_1,\ldots,r_{\ell_2}$ are polynomials in ${\bf x}$
    \item Each $f_i$ has the form
    \begin{equation}
        f_i({\bf x},\tilde{\bf x}) = \sum_{j=1}^{v} q_{i,j}({\bf x}) x_{n+j},
    \end{equation}
    where $q_{i,j}({\bf x})$ is a polynomial in $\bf x$.
\end{enumerate}
\end{define}
Note that if $v = 0$, ${\mathcal P}$ is a general polynomial set without any specification and so the problem degenerates. In the following discussion, we assume $v \ge 1$. 

Let $\mathcal{P}\subset \mathbb{F}_2[x_1, \ldots, x_{n+v}]$ be a polynomial set.

\begin{define}
Let
\begin{equation}\label{eq:projection}
\proj_{\x}(\zero({\mathcal P})) = \left\{ {\mathbf{x}} \in \mathbb{F}_2^n : \exists\  \tilde{\bf x} \in \mathbb{F}_2^v \text{ such that } ({\mathbf{x}}, \tilde{\bf x}) \in \zero({\mathcal P}) \right\}.
\end{equation}
Then $\proj_{\x}(\zero({\mathcal P}))$ is called the projection of $\zero({\mathcal P})$ onto the ${\bf x}$-variables.
\end{define}

\begin{remark}
For a characteristic set defined as Definition~\ref{def:CS}, we have $\zero(\trunc_{\bf x}(\T)) = \proj_{\bf x}(\zero(\T))$.
\end{remark}

\begin{define}\label{def:projection}
Let ${\rm zocd}(\mathcal{P})=\{\mathcal{T}_1,\ldots,\T_s\}$, then ${\rm zocd}(\mathcal{P})$ is called a zero $\x$-projection orthogonal characteristic decomposition of $\mathcal{P}$ if
	\begin{equation}\label{eq:projection_orthogonal}
\proj_{\x}(\zero(\mathcal{T}_i)) \cap \proj_{\x}(\zero(\mathcal{T}_j)) = \emptyset,
	\end{equation}
for any $1\le i<j\le s$, which is denoted by ${\rm zxocd}(\mathcal{P})=\{\mathcal{T}_1,\ldots,\T_s\}$.
\end{define}

\begin{define}
Let
\begin{equation}\label{eq:set_S}
    {\mathcal S} = \{ {\bf x}^* \in \mathbb{F}_2^n: ({\bf x}^*, {\bf 0}) \in \zero({\mathcal P}); ({\bf x}^*, \tilde{\bf x}) \notin \zero({\mathcal P}) \text{ for any }\tilde{\bf x} \ne {\bf 0}  \}.
\end{equation}
Then, $\mathcal{S}$ is called the feasible solution set of $\mathcal{P}$, denoted by ${\rm fss}(\mathcal{P})=\mathcal{S}$.
\end{define}


\begin{define}\label{eq:admissable}
    Let $\T$ be a characteristic set in $(\x,\tilde{\x})$. Then $\T$ is called admissible if $\T$ can be written in the form:
    \begin{align*}
        &{\mathcal T}'({\bf x})\\
        &x_{n+1}\\
        &x_{n+2}\\
        &\vdots\\
        &x_{n+v},
    \end{align*}
    where ${\mathcal T}'({\bf x})$ is a characteristic set in the ${\bf x}$-variables.
\end{define}

\begin{define}\label{def:admissiblezero}
Let $\T^*={\rm zxocd}(\mathcal{P})$. Then $\mathcal{X}$ is called the admissible zero set of $\mathcal{P}$ if
\begin{equation}\label{eq:X_adm}
    {\mathcal X} = \bigcup_{{\mathcal T} \in {\mathcal T}^*, {\mathcal T} \text{ is admissible}} \proj_{\x}(\zero({\mathcal T})).
\end{equation}
\end{define}


\begin{lemma}\label{lemma:necessary}
    ${\mathcal S} \subset {\mathcal X}$.
\end{lemma}
\begin{IEEEproof}
    If ${\mathcal S} = \emptyset$, the lemma is proved. Then, we assume ${\mathcal S} \ne \emptyset$. Let ${\bf x}^* \in {\mathcal S}$. By~\eqref{eq:set_S}, we have $({\bf x}^*, {\bf 0})\in\zero({\mathcal P})$. By Definitions \ref{def:projection} and \ref{def:admissiblezero}, we have $({\bf x}^*, {\bf 0}) \in \zero({\mathcal T})$ for a unique ${\mathcal T} \in {\mathcal T}^*$. To prove $x^*\in\mathcal{X}$, we only need to prove that ${\mathcal T}$ is admissible. We prove this by considering the following two cases.

    \emph{1) $x_{n+i}\in {\rm fvs}(\T)$ for some $1 \le i \le v$.} Without loss of generality, we assume that $x_{n+1}$ is a free variable. Then, by assigning $x_{n+1} = 1$, there exists $\tilde{\bf x} \ne \bf 0$ such that $({\bf x}^*,\tilde{\bf x})\in\zero({\mathcal P})$. By~\eqref{eq:set_S}, this contradicts the fact that ${\bf x}^* \in {\mathcal S}$.

    \emph{2) $x_{n+i}\notin {\rm fvs}(\T)$ for any $1 \le i \le v$.} Let ${\rm fvs}(\T)=\{x_{j_1},\ldots,x_{j_t}\}, \mbox{ where } 1\le j_1<\cdots< j_t\le n$.
    
    If $t=0$, i.e., ${\rm df}(\T) = 0$, then $\T$ has a unique zero $({\bf x}^*, {\bf 0})$, which implies that $\T$ is admissible.
    
    If $t>0$, i.e., ${\rm df}(\T)>0$, then $\mathcal T$ can be written as
    \begin{align*}
    &{\mathcal T}'({\bf x})\\
    &x_{n+1}+U_1 \\
    &x_{n+2}+U_2 \\
    &\vdots \\
    &x_{n+v}+U_v,
    \end{align*}
    where ${\mathcal T}'({\bf x})$ is a characteristic set in the ${\bf x}$-variables, and $U_k$ is a polynomial in ${\rm fvs}(\T)$. Assume that there exists an assignment of ${\rm fvs}(\T)$, denoted by $\phi$, such that $x_{n+i} \ne 0$ for some $1\le i\le v$. Let ${\bf x}(\phi)$ and  $\tilde{\bf x}(\phi)$ denote the value of $\bf x$ and $\tilde{\bf x}$ under the assignment $\phi$, respectively. Since $(\x(\phi),0)\in\zero(\T)$ and ${\rm fvs}(\T)$ only contains $\bf x$-variables, we obtain  $\tilde{\bf x}(\phi) = {\bf 0}$, which contradicts that $x_{n+i}\neq0$ for some $1\le i\le v$.
    Thus, $\tilde{\x}(\phi)= {\bf 0}$ for any $\phi$, which implies that
    $U_1 \equiv \cdots \equiv U_v \equiv 0$. Therefore, $\T$ is admissible. 

   Combining cases 1) and 2), we proved $x^*\in \mathcal{X}$, thus $\mathcal {S}\subset\mathcal{X}$.
\end{IEEEproof}

\begin{lemma}\label{lemma:sufficient}
     ${\mathcal X} \subset {\mathcal S}$.
\end{lemma}
\begin{IEEEproof}
    If ${\mathcal X} = \emptyset$, the lemma is proved. Then, we assume ${\mathcal X} \ne \emptyset$. Let ${\bf x}^* \in {\mathcal X}$. Since $\T^*={\rm zxocd}(\mathcal{P})$ and $({\bf x}^*,{\bf 0})\in\zero(\T)$ for some $\mathcal T \in {\mathcal T}^*$, we have the fact that $({\bf x}^*,{\bf 0})$ is a solution to ${\mathcal P} = 0$. To prove ${\bf x}^* \in {\mathcal S}$, we only need to prove that $({\bf x}^*,\tilde{\bf x})$ is not a solution to ${\mathcal P} = 0$ for any $\tilde{\bf x} \ne {\bf 0}$. Again since $\T^*={\rm zxocd}(\mathcal{P})$, we have $({\bf x}^*,\tilde{\bf x})\notin \zero(\T')$ for any $\T'\neq\T$.
    By Definition \ref{def:admissiblezero}, $\T$ is admissible. Then, by Definition \ref{eq:admissable}, $({\bf x}^*,\tilde{\bf x}) \notin \zero({\mathcal T})$ for any $\tilde{\bf x} \in (\mathbb{F}_2^v \setminus \{{\bf 0}\})$.
    Thus, we proved that $({\bf x}^*, {\bf 0}) \in \zero({\mathcal P})\mbox{ and } ({\bf x}^*, \tilde{\bf x}) \notin \zero({\mathcal P}) \text{ for any }\tilde{\bf x} \ne {\bf 0} $, which implies $\x^*\in \mathcal{S}$.
    Therefore, $\mathcal{X}\subset \mathcal{S}$.
    %
%
%
%
\end{IEEEproof}

Recall that Problem~\ref{prob:P1} is equivalent to Problem~\ref{prob:P1a}.
Then the two feasible sets of Problem~\ref{prob:P1} and Problem~\ref{prob:P1a} are equivalent.
%
%
%
For Problem \ref{prob:P1a}, let 
\begin{align}\label{P4Problem}
\P=\{{\bm \zeta}_i \mathbf{A}_i({\bf x}),\ h_j(\mathbf{x}),\ i = 1,\ldots,m_1,\ j=1,\ldots,m_2\}, 
\end{align}
where $\tilde{\x}=({\bm \zeta}_1,\ldots,{\bm \zeta}_{m_1})$, and let $\T^*={\rm zxocd}(\P)$, then by Lemmas~\ref{lemma:necessary} and \ref{lemma:sufficient}, we obtain the following theorem.

\begin{theorem}\label{theorem:main}
The feasible set of Problem \ref{prob:P1a} (Problem \ref{prob:P1}) is equal to $\mathcal{X}$, which is defined in \eqref{eq:X_adm}.
\end{theorem}

Note that $\P$ in \eqref{P4Problem} is a polynomial set in $(\x,\tilde{\x})$, and $\tilde{\x}$-linear (Definition~\ref{def:x_tilde_linear_poly}). 
At this point, the key problem is how to find the feasible set of Problem~\ref{prob:P1a} by deriving $\mathcal{X}$, that is, how to derive ${\rm zxocd}(\P)$. To solve this problem, we present Algorithm~\ref{BCS-ALG} based on Procedure~\ref{proc:BCS}, which is called the \emph{BCS algorithm}. The BCS algorithm is implemented by only two basic operations (initial decomposition and substitution). Thus, we will prove in Theorem~\ref{theorem:BCS} that the BCS algorithm can derive ${\rm zxocd}(\P)$.



\begin{algorithm}[!ht]
	\small	
    \label{BCS-ALG}
	\caption{BCS Algorithm}
	\KwIn{$\tilde{\bf x}$-linear polynomial set ${\mathcal P}$}
	\KwOut{${\rm zxocd}(\P)={\mathcal T}^* =\{{\mathcal T}_1,\ldots,{\mathcal T}_s\}$}
    ${\mathcal P}^* \leftarrow \{{\mathcal P}\}$, ${\mathcal T}^* \leftarrow \emptyset$\;
	\While{$\P^* \ne \emptyset$}
	{
        Arbitrarily select a branch ${\mathcal Q} \in {\mathcal P}^*$ and let ${\mathcal P}^* \leftarrow {\mathcal P}^* \setminus \{{\mathcal Q}\}$\; 
        ${\mathcal T} \leftarrow \emptyset$;\tcp*[f]{characteristic set}\\
        ${\mathcal Q}^* \leftarrow \emptyset$;\tcp*[f]{newly generated branches}\\  
        \lIf{$1 \in {\mathcal Q}$}{\textbf{goto} Line~\ref{BCS-line-Pstar-update}}
        \While{$({\mathcal Q} \setminus \{0\}) \ne \emptyset$}{
                Let $f$ be a non-constant polynomial chosen from ${\mathcal Q}$, and ${\rm can}(f) = Ix_c+U$\;\nllabel{BCS-line-choose}
                \eIf{$I = 1$}
                {
                    ${\mathcal Q} \leftarrow {\mathcal Q} \setminus \{f\}$\;                     
                }
                {      
                    ${\mathcal Q}_0 \leftarrow ({\mathcal Q}\setminus \{f\}) \cup {\mathcal T} \cup \{I,U\}$\;\nllabel{BCS-line-Q0}
                    ${\mathcal Q}^* \leftarrow {\mathcal Q}^* \cup \{{\mathcal Q}_0\}$;\tcp*[f]{record the branch generated by $I = 0$}\\  
                    ${\mathcal Q} \leftarrow ({\mathcal Q} \setminus \{f\}) \cup \{I+1\}$;\tcp*[f]{continue with the branch generated by $I=1$}\\                        
                }
                ${\mathcal Q} \leftarrow {\rm subs}(x_c=U,{\mathcal Q})$,
                ${\mathcal T} \leftarrow {\rm subs}(x_c=U,{\mathcal T})$,
                ${\mathcal T} \leftarrow {\mathcal T} \cup \{x_c+U\}$\;\nllabel{BCS-line-QTupdate}
                \lIf{$1 \in {\mathcal T}$ or $1 \in {\mathcal Q}$}{$\T \leftarrow \emptyset$;~\textbf{break}}            
        }
        ${\mathcal P}^* \leftarrow {\mathcal P}^* \cup {\mathcal Q}^*$\;\nllabel{BCS-line-Pstar-update}
        \lIf{$\T \ne \emptyset$}{${\mathcal T}^* \leftarrow {\mathcal T}^* \cup \{{\mathcal T} \setminus \{0\}\}$}
	}
\end{algorithm}
    
\begin{proposition}\label{fact1}
In Algorithm \ref{BCS-ALG}, the initial $I$ for any selected polynomial in Line~\ref{BCS-line-choose} depends only on the ${\bf x}$-variables.
\end{proposition}
\begin{IEEEproof}
    An initial $I$ that depends on $\tilde{\bf x}$-variables may be chosen only if there exists a term $x_{n+i} x_{n+j} f$ for $i \ne j$ in some polynomial, where $f$ can be any polynomial in $x_1,\ldots,x_{n+v}$. Note that there is no such a term in the polynomial set $\mathcal P$. In the BCS algorithm, the only operation that can generate new (non-constant) terms is the substitution. When substituting $x_c = U$, $U$ is a polynomial in $x_1,\ldots,x_{c-1}$. It follows that no variable $x_i$, $i \le n$, can be replaced by a polynomial involving any $x_{n+j}$, $j \ge 1$. Consequently, no terms satisfying the form $x_{n+i} x_{n+j} f$ can arise in any polynomial involved in the BCS algorithm. This guarantees that any initial $I$ chosen in the BCS algorithm depends only on the ${\bf x}$-variables.
\end{IEEEproof}

\begin{theorem}\label{theorem:BCS}
    Algorithm~\ref{BCS-ALG} can derive ${\rm zxocd}(\P)$.
\end{theorem}
\begin{IEEEproof}
Based on Procedure~\ref{proc:BCS}, the output of Algorithm \ref{BCS-ALG}, $\T^*$, is a zero orthogonal characteristic decomposition of $\mathcal{P}$, i.e., ${\rm zocd}(\P)=\T^*=\{\T_1,\ldots,\T_s\}$. Then, we only need to prove that $\T^*$ satisfies~\eqref{eq:projection_orthogonal}.
We will prove this theorem from a high-level perspective of the BCS algorithm, considering only those steps that affect the zero sets. A proof that strictly follows each step of the BCS algorithm would be cumbersome. 
    
Let $\kappa$ be the index for the iterations in the BCS algorithm, where we regard every round of the second while-loop in Algorithm~\ref{BCS-ALG} as an iteration. The BCS algorithm decomposes a polynomial set $\P$ into a set of branches ${\mathcal B}_\kappa^*$ at every iteration~$\kappa$ such that $\bigcup_{{\mathcal B} \in {\mathcal B}_\kappa^*} \zero({\mathcal B}) = \zero(\P)$. At iteration~$\kappa+1$, the new set of branches ${\mathcal B}_{\kappa+1}^*$ is generated by the following steps: 
\begin{enumerate}
    \item Choose a branch ${\mathcal B} \in {\mathcal B}_\kappa^*$ and let ${\mathcal B}_{\kappa+1}^* \leftarrow {\mathcal B}_\kappa^* \setminus \{\mathcal B\}$.
    \item Generate two branches ${\mathcal B}_1 \leftarrow {\mathcal B} \cup \{I\}$ and ${\mathcal B}_2 \leftarrow {\mathcal B} \cup \{I + 1\}$, where $I$ is the initial of the chosen polynomial. (Strictly speaking, $\mathcal{Q}_0$ updated by Line~\ref{BCS-line-Q0} of Algorithm~\ref{BCS-ALG} is the branch ${\mathcal B}_1$, and $\mathcal{Q} \cup \mathcal{T}$ updated by Line~\ref{BCS-line-QTupdate} of Algorithm~\ref{BCS-ALG} is the branch $\mathcal{B}_2$.)
    \item ${\mathcal B}_{\kappa+1}^* \leftarrow {\mathcal B}_{\kappa+1}^* \cup \{{\mathcal B}_1\} \cup \{{\mathcal B}_2\}$.
\end{enumerate}

This way, for any two branches ${\mathcal B},{\mathcal B}' \in {\mathcal B}_\kappa^*$, there exists an initial $I'$ such that $\zero({\mathcal B}) \subset \zero(I')$ and $\zero({\mathcal B}') \subset \zero(I'+1)$. By Proposition \ref{fact1}, this initial $I'$ depends only on the ${\bf x}$-variables. Thus, $\proj_{\bf x}(\zero({\mathcal B})) \cap \proj_{\bf x}(\zero({\mathcal B}')) = \emptyset$ for any distinct ${\mathcal B},{\mathcal B}' \in {\mathcal B}_\kappa^*$.

Let $\kappa_{\rm last}$ be the index for the last iteration. Since every $\T \in \T^*$ is obtained, via reduction, from a distinct branch in ${\mathcal B}_{\kappa_{\rm last}}^*$, the proof of this theorem is complete.
\end{IEEEproof}

\begin{algorithm}[!ht]\label{alg:BCSFR}
	\small
	\caption{BCSFR}
	\KwIn{$\tilde{\bf x}$-linear polynomial set ${\mathcal P}$}
	\KwOut{${\mathcal T}^* =\{{\mathcal T}_1,\ldots,{\mathcal T}_s\}$ such that ${\rm fss}(\P) = \zero(\T^*)$, where $\T_i$ is a characteristic set in $\bf x$}
    ${\mathcal P}^* \leftarrow \{{\mathcal P}\}$,
    ${\mathcal T}^* \leftarrow \emptyset$\;
	\While{$\P^* \ne \emptyset$}
	{
        Arbitrarily select a branch ${\mathcal Q} \in {\mathcal P}^*$ and let ${\mathcal P}^* \leftarrow {\mathcal P}^* \setminus \{{\mathcal Q}\}$\;
        ${\mathcal T} \leftarrow \emptyset$\;      
        \lIf{$1 \in \mathcal{Q}$}{\textbf{goto} Line~\ref{line:last}}
        \While{$({\mathcal Q} \setminus \{0\}) \ne \emptyset$}{ 
                \While{$|\{f = Ix_c+U \in \mathcal{Q}: I = 1\}|>0$} {
                $f \leftarrow \texttt{Choose}(\{f = Ix_c+U \in \mathcal{Q}: I = 1\},0)$ and let ${\rm can}(f) = Ix_c+U$\;
                $\mathcal{Q} \leftarrow{{\rm subs}(x_c = U,\mathcal{Q}\backslash f)}, \T \leftarrow{{\rm subs}(x_c = U,\T)}, \T \leftarrow \T \cup\{f\}$\;
                \If{$1 \in {\mathcal T}$ or $1 \in {\mathcal Q}$ or $x_k+1 \in {\mathcal T}$ for some $k \ge n+1$}
                {$\T \leftarrow \emptyset$;~\textbf{goto} Line~\ref{line:last}\;}
                }
                \lIf{$({\mathcal Q} \setminus \{0\}) = \emptyset$}{\textbf{goto} Line~\ref{line:last}}
                ${\mathcal Q} \leftarrow {\mathcal Q} \setminus \{0\}$\;
                \eIf{$ |\{f \in \mathcal{Q}:{\rm cls}(f) \le n\}|>0$} {$f \leftarrow \texttt{Choose}(\{f \in \mathcal{Q}:{\rm cls}(f) \le n\},1)$\;}
                {$k^* \leftarrow{n+v}$\;
                \While{$k^* > n$}{
                    \lIf{$|\{f\in \T: {\rm cls}(f) = k^*\}| > 0$}{$k^* \leftarrow{k^* - 1}$}
                    \lElseIf{$|\{f\in \mathcal{Q}:{\rm cls}(f) = k^*\}|>0$}{$f \leftarrow \texttt{Choose}({\mathcal Q},2)$; \textbf{break}}
                    \lElse{$\T\leftarrow \emptyset$;
                    \textbf{goto} Line~\ref{line:last}}
                }}   
                Let ${\rm can}(f) = Ix_c + U$\;
                \eIf{$U = 1$ and $I = \prod_{i=1}^{j} x_{k_i}$ for some $j$ and some $\{k_1,\ldots,k_j\}$}{
                    ${\mathcal Q} \leftarrow {\mathcal Q} \setminus \{f\}$\; 
                    \For{$i=1,\ldots,j$}
                    {
                        ${\mathcal Q} \leftarrow {\rm subs}(x_{k_i}=1,{\mathcal Q})$,
                        ${\mathcal T} \leftarrow {\rm subs}(x_{k_i}=1,{\mathcal T})$,
                        ${\mathcal T} \leftarrow {\mathcal T} \cup \{x_{k_i}+1\}$\;
                    }
                }
                {      
                    \eIf{$I = 1$}{
                    ${\mathcal Q} \leftarrow {\mathcal Q} \setminus \{f\}$\;
                    }
                    {
                    \lIf{$U \ne 1$}{
                    $\P^* \leftarrow \P^* \cup \{({\mathcal Q}\setminus \{f\}) \cup {\mathcal T} \cup \{I,U\}\}$
                    }  
                    ${\mathcal Q} \leftarrow ({\mathcal Q} \setminus \{f\}) \cup \{I+1\}$\;
                    }      
                }
                ${\mathcal Q} \leftarrow {\rm subs}(x_c=U,{\mathcal Q})$,
                ${\mathcal T} \leftarrow {\rm subs}(x_c=U,{\mathcal T})$,
                ${\mathcal T} \leftarrow {\mathcal T} \cup \{x_c+U\}$\;
                \If{$1 \in {\mathcal T}$ or $1 \in {\mathcal Q}$ or $x_k+1 \in {\mathcal T}$ for some $k \ge n+1$}
                {$\T \leftarrow \emptyset$;~\textbf{break}\;}            
        }
        \lIf{$\mathcal T$ is admissible\nllabel{line:last}}{${\mathcal T}^* \leftarrow {\mathcal T}^* \cup \{{\rm Trunc}_{\bf x}({\mathcal T} \setminus \{0\})\}$;\tcp*[f]{Definition~\ref{def:trunc}}}
	}
\end{algorithm}

\subsection{BCSFR Algorithm}\label{subsec:alg}
According to Definition~\ref{def:admissiblezero}, only admissible characteristic sets are required. Consequently, to enhance algorithmic efficiency, it is intuitive to identify and prune the branches incapable of producing admissible characteristic sets at an early stage of the algorithm. To this end, we develop the BCSFR algorithm, presented in Algorithm~\ref{alg:BCSFR}. 
The function $\texttt{choose}$ used in Algorithm~\ref{alg:BCSFR} to choose a \emph{non-constant polynomial} is defined below.

\Fn{\Choose{$\mathcal{Q},i$}}{
        \lIf{$i=1$}{
            \Return{$f \in (\mathcal{Q} \setminus \{0,1\})$ chosen by 1st-CSO}
        }
        \lElseIf{$i=2$}{
            \Return{$f \in (\mathcal{Q} \setminus \{0,1\})$ chosen by 2nd-CSO}
        }
        \lElse{
            \Return{$f \in (\mathcal{Q} \setminus \{0,1\})$ chosen randomly}
        }
    }

\begin{remark}
    Note that in Algorithm~\ref{alg:BCSFR}, we also make some modifications to improve the computational efficiency of the algorithm. For example, we speed up the initial decomposition for the polynomials of special forms. Consider $f = x_1x_2x_3 + 1$. We can immediately know the final decomposition result of $f$ is $\{x_1+1,x_2+1,x_3+1\}$, which means that we do not need to perform initial decompositions step by step for such polynomials. These modifications are general for CS-based algorithms and are not dedicated to the considered linear coding problems.
\end{remark}

An example is given in Appendix~\ref{appendix:BCSFR_example} to elaborate on the improvements in Algorithm~\ref{alg:BCSFR}.

Considering the special structure of Problem~\ref{prob:P1}, where $m_1$ full-rank constraints are required to hold simultaneously, we can use the BCSFR algorithm in an incremental way to address these full-rank constraints progressively. To be more precise, we can form $m_1$ $\tilde\x$-linear polynomial sets $\P_1,\ldots,\P_{m_1}$ respectively according to the $m_1$ full-rank constraints (the non-rank constraints can be arranged to any $\P_i$). We first use the BCSFR to solve $\P_1$ and obtain a collection $\mathcal{T}_1^*$ of characteristic sets. Then we use the BCSFR algorithm $|\mathcal{T}_1^*|$ times, each solving the union of a characteristic set in $\mathcal{T}_1^*$ and $\P_2$, and recording the results in $\mathcal{T}_2^*$. The processing of $\P_3,\ldots,\P_{m_1}$ are the same as $\P_2$. Finally, we obtain $\mathcal{T}_{m_1}^*$ as the result. In this way, the problem size handled at each invocation is significantly reduced.

We present the Incremental BCSFR (Inc-BCSFR) in Algorithm~\ref{alg:R-BCSFR}. Here, we provide a brief explanation for its correctness. The feasible set of Problem~\ref{prob:P1} can be written as $\bigcap_{i = 1}^{m_1} {\rm fss}(\P_i)$. After the $i$-th for-loop of the Inc-BCSFR algorithm (line 3), we obtain $\bigcap_{j = 1}^{i} {\rm fss}(\P_j) = \zero(\T_i^*)$. Thus, we can obtain the feasible set of Problem~\ref{prob:P1} after running all the for-loops. 

\begin{algorithm}[h]\label{alg:R-BCSFR}
	\caption{Incremental BCSFR (Inc-BCSFR)}
	\KwIn{$\tilde{\bf x}$-linear polynomial sets ${\mathcal P}_1,\ldots,{\mathcal P}_{m_1}$}
	\KwOut{${\mathcal T}^{*} =\{{\mathcal T}_1,\ldots,{\mathcal T}_s\}$ such that $\bigcap_{i=1}^{m_1} {\rm fss}(\P_i) = \zero(\T^*)$, where $\T_i$ is a characteristic set in $\bf x$}
    ${\mathcal T}_0^* \leftarrow \{0\}$\;
    ${\mathcal T}_1^* \leftarrow \emptyset,\ldots,{\mathcal T}_{m_1}^* \leftarrow \emptyset$\;
    \For{$i=1,\ldots,m_1$}
    {
        \ForEach{${\mathcal T} \in {\mathcal T}_{i-1}^*$}
        {
            Let $\mathcal{F}^*$ be the output of the BCSFR algorithm for solving $\P_i \cup {\mathcal T} $\;
            ${\mathcal T}_i^* \leftarrow {\mathcal T}_i^* \cup {\mathcal F}^*$\;
        }
    }
    ${\mathcal T}^* \leftarrow {\mathcal T}_{m_1}^*$\;
\end{algorithm}

\section{Application to Linear Network Codes}
\label{sec-appl-nc}
In this section, we show the application of the proposed algorithm to network communication scenarios employing LNCs. We first present the basic mathematical models of networks and LNCs. Interested readers are referred to~\cite{yeung2008information} for further details on network coding. Then, we formulate an LNC problem with certain practical constraints, which satisfies the form of Problem~\ref{prob:P1}. Lastly, we provide an illustrative example to show our solution.

\subsection{Fundamentals of Linear Network Codes}
We consider a (communication) network model consisting of the following components:
\begin{enumerate}
	\item  Directed acyclic graph (DAG) $\mathcal{G}$: A network is represented by a DAG $\mathcal{G} = (\mathcal{V},\mathcal{E})$, where $\mathcal{V}$ is the set of nodes and $\mathcal{E}$ is the set of edges, which represent point-to-point channels.
	\item Source node $s$: Let node $s \in \mathcal{V}$ be the unique source that generates a message consisting of $\omega$ symbols in $\mathbb{F}_2$, represented by a row vector $\mathbf{m} \in \mathbb{F}_2^{\omega}$.\footnote{In real-world applications, the source node generates $\omega$ \emph{packets} rather than \emph{symbols}. Each packet can be viewed as a column $T$-vector consisting of $T$ bits. In general, $T$ is a large number, say $8192$. However, this difference does not affect the mathematical analysis in this paper. For simplicity, we adopt the assumption that each packet has only one symbol.}
	\item User node set $\mathcal{U}$: Let $\mathcal{U} \subset \mathcal{V}$ be the set of user nodes. Each user node $u \in \mathcal{U}$ is required to recover all the symbols generated by the source node,~i.e., the $\omega$-vector $\mathbf{m}$.
\end{enumerate}
We refer to the triple $({\mathcal G},s,{\mathcal U})$ as a network.

In the network $({\mathcal G},s,{\mathcal U})$, parallel edges between a pair of nodes is allowed. We assume without loss of generality that all the edges in the network have unit capacity, i.e., one symbol in $\mathbb{F}_2$ can be transmitted on each edge by one use of the edge. For node~$t \in{\mathcal V}$, denote its sets of incoming and outgoing edges by ${\rm In}(t)$ and ${\rm Out}(t)$, respectively. We use $(v,t)$ to represent the edge from node~$v$ to node~$t$ if there exists only one edge between them.

For source node $s$, assume that the $\omega$ source symbols are injected by $\omega$ \textit{imaginary} incoming edges, and let ${\mathcal E}_{0}$ denote the set of the imaginary edges. We define the augmented edge set $\tilde{\mathcal E}= {\mathcal E}\cup{\mathcal E}_0$, and the augmented graph $\tilde{\mathcal G}=({\mathcal V},\tilde{\mathcal E})$. Unless stated otherwise, ${\rm In}(t)$ and ${\rm Out}(t)$ are taken with respect to the augmented graph $\tilde{\mathcal G}$.

Let ${\mathcal E}_{\rm pair}$ be the set of adjacent edge pairs,
\begin{equation}\label{eq:E_pair}
{\mathcal E}_{\rm pair} = \bigcup_{t\in{\mathcal V}} {\rm In}(t)\times{\rm Out}(t).
\end{equation}
\begin{define}[Linear network code]\label{def:LNC}
An LNC on $\tilde{\mathcal G}$ over $\mathbb{F}_2$ is an assignment of encoding coefficients,
\begin{equation}\nonumber
{\mathscr C} = \big(g_{e_0,e_1}\big)_{(e_0,e_1)\in{{\mathcal E}_{\rm pair}}},\quad g_{e_0,e_1}\in\mathbb{F}_2.
\end{equation}
For each node $t \in{\mathcal V}$, let
\begin{equation}\nonumber
{\bf G}_t = [g_{e_0,e_1}]_{e_0 \in {\rm In}(t),\, e_1 \in {\rm Out}(t)}
\in \mathbb{F}_2^{|{\rm In}(t)| \times |{\rm Out}(t)|},
\end{equation}
which is called the local encoding matrix for node $t$. Equivalently, 
\begin{equation}\nonumber
    {\mathscr C} = ({\bf G}_{t})_{t\in{\mathcal{V}}}.
\end{equation}
\end{define}

\begin{remark}
    Note that in (\ref{eq:E_pair}), ${\rm In}(t)\times{\rm Out}(t)= \emptyset$ if ${\rm In}(t) = \emptyset$ or ${\rm Out}(t) = \emptyset$.
\end{remark}

\begin{remark}
    The matrix structure of ${\bf G}_t$ implicitly assumes an ordering among the incoming and outgoing edges of node $t$.
\end{remark}

An \emph{admissible} LNC $\mathscr C$ should make every user node be able to recover \emph{all} the source symbols. Before formally defining admissible LNCs, we need to formally define the observations available at the user nodes. 

Given a code ${\mathscr C}$, define the \emph{global encoding vector} $\mathbf{f}_e\in\mathbb{F}_2^{\omega}$ for every edge $e\in\tilde{\mathcal E}$ as a column $\omega$-vector satisfying:
\begin{enumerate}
	\item (\textit{Imaginary edges}) For the source node $s$ and its $i$-th imaginary incoming edge $e \in{\mathcal E}_{0}$,
	\begin{equation}\nonumber
	\mathbf{f}_{e} = \mathbf{e}_{i},
	\end{equation}
    where ${\bf e}_i$ is a zero column $\omega$-vector except that the $i$-th position is one.
	\item (\emph{Real edges}) For any real edge $e\in{\rm Out}(t)\cap{\mathcal E}$,
	\begin{equation}\label{eq:genc_kernel}
	\mathbf{f}_e
	= \sum_{e_0\in{\rm In}(t)} g_{e_0,e} \mathbf{f}_{e_0}.
	\end{equation}
\end{enumerate}

\begin{remark}
    In the recursive computation of $\mathbf{f}_{e}$, an order of edges is implied. When computing $\mathbf{f}_{e}$, $e\in{\rm Out}(t)$, the global encoding vectors $\mathbf{f}_{e_0}$, $e_0\in{\rm In}(t)$ should be computed first. Such orders exist because ${\mathcal G}$ is acyclic.
\end{remark}

Note that $\{\mathbf{e}_1,\ldots,\mathbf{e}_\omega\}$ form the standard basis of $\mathbb{F}_2^\omega$. For each $e\in\tilde{\mathcal E}$, the symbol carried by edge $e$, denoted as $y_e \in \mathbb{F}_2$, can be written as
\begin{equation}\nonumber
y_e = \mathbf{m} \mathbf{f}_e,
\end{equation}
where $\mathbf{m} \in \mathbb{F}_2^{\omega}$ is the row vector of all source symbols.

Using the global encoding vectors $\mathbf{f}_{e}$, the observations at user nodes can be formulated. For a user node $u \in{\mathcal U}$, collect the global encoding vectors for the edges in ${\rm In}(u)$ into the matrix,
\begin{equation}\label{eq:receive_matrix}
{\bf F}_u = \big[\mathbf{f}_e \big]_{e\in{\rm In}(u)}
\in \mathbb{F}_2^{\omega\times |{\rm In}(u)|},
\end{equation}
where the columns are ordered according to any fixed ordering of ${\rm In}(u)$.
Let $\mathbf{y}_u = (y_e)_{e\in{\rm In}(u)}$ denote the received row vector at user node $u$. Then,
\begin{equation}\nonumber
\mathbf{y}_u = \mathbf{m} {\bf F}_u.
\end{equation}
Clearly, node $u$ can recover $\bf m$ if and only if ${\bf F}_u$ has full row rank.

\begin{define}[Admissible LNC for $({\mathcal G},s,{\mathcal U})$]\label{def:admissible}
An LNC~${\mathcal C}$ over $\mathbb{F}_2$ is \emph{admissible} for the network $({\mathcal G},s,{\mathcal U})$ if for every user node $u\in{\mathcal U}$,
\begin{equation}\nonumber
{\rm rank}\left({\bf F}_u\right)=\omega.
\end{equation}
\end{define}

\begin{example}
    Consider the network in Fig.~\ref{fig:butterfly}, where the source node $s$ is to transmit $\omega =2$ symbols $b_1$ and $b_2$ to two user nodes $u_1$ and $u_2$. The local encoding matrices are:
    \begin{equation*}
        \begin{aligned}
            &{\bf G}_{s} = \begin{bmatrix}
                1 & 0\\
                0 & 1
            \end{bmatrix},\quad {\bf G}_{t_1} = {\bf G}_{t_2} = \begin{bmatrix}
                1 & 1
            \end{bmatrix},\quad
            {\bf G}_{t_3} = \begin{bmatrix}
                1 \\ 1
            \end{bmatrix},\quad {\bf G}_{t_4} = \begin{bmatrix}
                1 & 1
            \end{bmatrix}.
        \end{aligned}
    \end{equation*}
    Using such local encoding matrices, one can obtain the global encoding vectors, which are shown near the corresponding edges. This way, node $u_1$ receives $b_1$ and $b_1 + b_2$, and node $u_2$ receives $b_2$ and $b_1 + b_2$. Obviously, both nodes $u_1$ and $u_2$ can recover $b_1$ and $b_2$.
\end{example}

\begin{figure}[!t]
	\centering
	\includegraphics[width=3in]{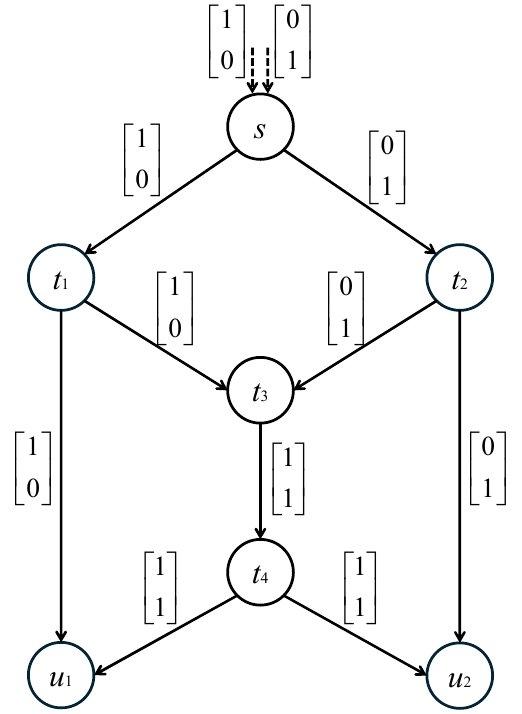}
	\caption{An example of an admissible LNC, where $\omega = 2$ and $\mathcal{U} = \{u_1,u_2\}$. The global encoding vectors are shown in the figure.}
	\label{fig:butterfly}
\end{figure}

\subsection{Problem Formulation}\label{subsec:appl_nc_formulation}
Let $\hat{\mathcal{V}} \subset \mathcal{V}$ be the set of nodes for which the local encoding matrices need to be designed, and let
\begin{equation*}
    \hat{\mathcal E}_{\rm pair} = \bigcup_{t\in\hat{\mathcal V}} {\rm In}(t)\times{\rm Out}(t).
\end{equation*}
For $t \in \mathcal{V}\setminus \hat{\mathcal{V}}$, its local encoding matrix is a constant matrix. Such nodes are referred to as \emph{constant nodes}.\footnote{The introduction of constant nodes is often attributed to the following reasons: 1) Some nodes may possess no outgoing edges, resulting in their local encoding matrices being empty matrices. 2) Some local encoding matrices are imposed to be fixed due to specific constraints or to reduce the number of variables in the optimization problem.} For every node $t \in \hat{\mathcal{V}}$, its local encoding matrix ${\bf G}_t$ has $|{\rm In}(t)| \times |{\rm Out}(t)|$ encoding coefficients which need to be optimized. In most literature on LNCs, all encoding coefficients can be independently designed. However, certain constraints may be desirable due to practical considerations. These constraints can be included into the non-rank constraints in Problem~\ref{prob:P1}. Except for constant nodes, we consider three types of nodes to be optimized:\footnote{Besides the constraints we consider, it is straightforward to include other types of constraints into the problem.}
\begin{enumerate}
    \item \emph{General node:} Such a node $t$ can perform linear combinations of the received symbols and no specific constraint is imposed to ${\bf G}_t$.
    \item \emph{Routing node:} Such a node $t$ can only perform routing, i.e., choose one of the received symbols carried by the incoming edges and forward it on an outgoing edge. In other words, the Hamming weight of every column in ${\bf G}_t$ is exactly $1$.
    \item \emph{Broadcast node:} Such a node $t$ can only broadcast symbols on its outgoing edges, i.e., transmit the same symbol on its outgoing edges. In other words, the columns in ${\bf G}_t$ are the same.
\end{enumerate}

We now formulate an \emph{LNC problem} in the form of Problem~\ref{prob:P1}. Every encoding coefficient $g_{e_0,e_1}$, $(e_0,e_1)\in\hat{\mathcal{E}}_{\rm pair}$, is represented by a binary variable $x_i$. Note that $i \in \{1,\ldots,n\}$, where $n = \sum_{t \in \hat{\mathcal{V}}} |{\rm In}(t)| \times |{\rm Out}(t)|$.\footnote{To avoid introducing additional notation, we keep using $n$ for the number of variables. The value of $n$ depends on the particular problem under consideration.} Let $\varphi: \hat{\mathcal{E}}_{\rm pair} \to \{1,\ldots,n\}$ be a bijection. Let ${\bf F}_u({\bf x})$ denote the matrix defined by (\ref{eq:receive_matrix}) with respect to ${\bf x} = (x_1,\ldots,x_n)$. Let $\mathcal{V}_{\rm g}$, $\mathcal{V}_{\rm b}$, and $\mathcal{V}_{\rm r}$ be the sets of general nodes, broadcast nodes, and routing nodes, respectively. Note that $\hat{\mathcal{V}} = \mathcal{V}_{\rm g} \cup \mathcal{V}_{\rm b} \cup \mathcal{V}_{\rm r}$. We have the following constraints for the LNC problem:
\begin{enumerate}
    \item \emph{Admissibility constraints: } For $u \in \mathcal{U}$,
    \begin{equation}\label{eq:rank_constraint_LNC}
        {\rm rank}({\bf F}_u({\bf x})) = \omega.
    \end{equation}
    \item \emph{Routing constraints:} For $t \in {\mathcal{V}_{\rm r}},~e \in {\rm Out}(t)$,
    \begin{equation}\label{eq:routing_constraint_LNC}
        \begin{aligned}
            &x_{\varphi(i,e)} x_{\varphi(j,e)} = 0,\quad i,j \in {\rm In}(t),~i\ne j,\\
            & \sum_{i \in {\rm In}(t)} x_{\varphi(i,e)} + 1 = 0.
        \end{aligned}
    \end{equation}
    \item \emph{Broadcast constraints:} For $t \in {\mathcal{V}_{\rm b}},~e \in {\rm In}(t)$,
    \begin{equation}\label{eq:broadcast_constraint_LNC}
        x_{\varphi(e,i)} = x_{\varphi(e,j)},\quad i,j \in {\rm Out}(t),~i\ne j.
    \end{equation}
\end{enumerate}

We can rearrange the above constraints into the two types of constraints in Problem~\ref{prob:P1}. The admissibility constraints provided above are consistent with the full-rank constraints in Problem~\ref{prob:P1}. The routing constraints and the broadcast constraints should be classified into the non-rank constraints in Problem~\ref{prob:P1}. Here, we do not delve into the objective function $f({\bf x})$ in Problem~\ref{prob:P1}. A practical objective function will be discussed in Sec.~\ref{sec-experiment}. For such an optimization problem, the feasible set can be explicitly obtained using the BCSFR/Inc-BCSFR algorithm.

If an exhaustive search is used to solve the LNC problem, the size of the search space is
\begin{equation}\label{eq:exhaustive_sparce}
    \prod_{t_1 \in \mathcal{V}_{\rm g}} 2^{|{\rm In}(t_1)|\times|{\rm Out}(t_1)|} \times \prod_{t_2 \in \mathcal{V}_{\rm r}} (|{\rm In}(t_2)|)^{|{\rm Out}(t_2)|} \times \prod_{t_3 \in \mathcal{V}_{\rm b}} 2^{|{\rm In}(t_3)|}.
\end{equation}

\subsection{An Illustrative Example}\label{subsec:example_lnc}
\begin{figure}[!t]
	\centering
	\includegraphics[width=5 in]{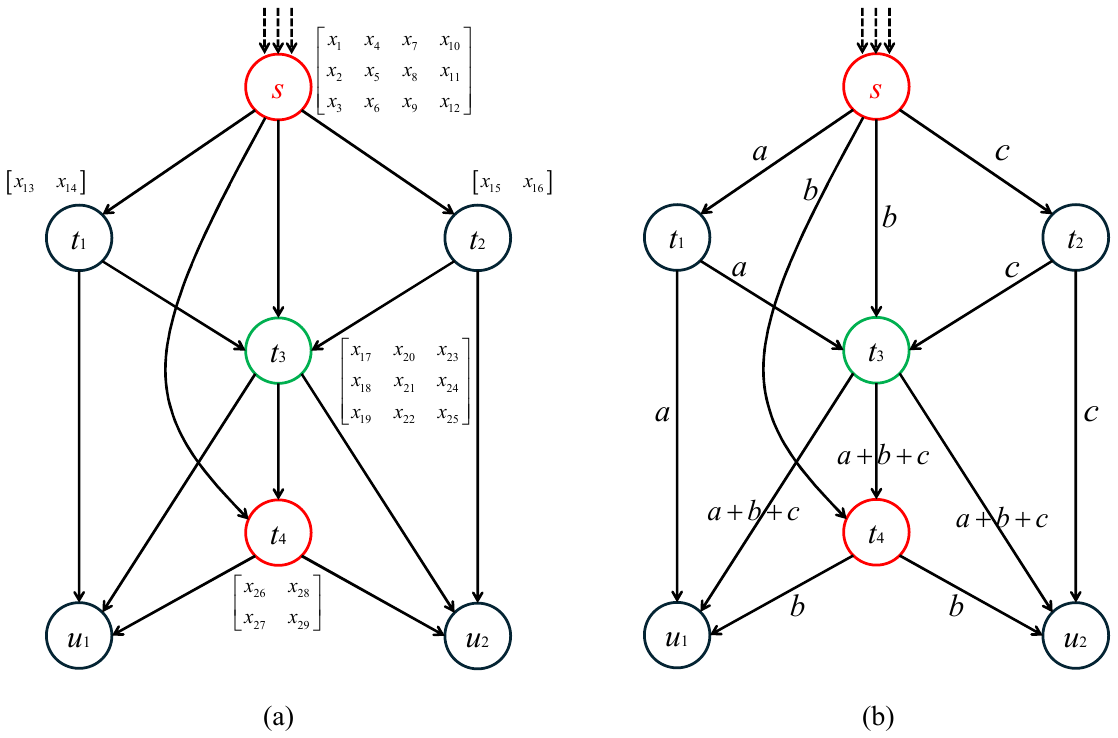}
	\caption{(a) A network $(\mathcal{G},s,\mathcal{U})$, where $\omega = 3$ and $\mathcal{U} = \{u_1,u_2\}$. The nodes highlighted in red and green are routing and broadcast nodes, respectively. The local encoding matrices represented by $x_1,\ldots,x_n$ are given. (b) An admissible LNC for the network in (a), illustrated by the expressions of the symbols carried by the edges.}
	\label{fig:LNC_example}
\end{figure}

Consider the network as shown in~Fig.~\ref{fig:LNC_example}(a), where $\mathcal{U} = \{u_1,u_2\}$, $\mathcal{V}_{\rm g} = \{t_1,t_2\}$, $\mathcal{V}_{\rm r} = \{s,t_4\}$, and $\mathcal{V}_{\rm b} = \{t_3\}$. Three symbols $a$, $b$, and $c$ are generated at the source node $s$, and are to be transmitted to two user nodes $u_1$ and $u_2$. The local encoding kernels are represented by the binary variables ${\bf x} = (x_1,\ldots,x_{29})$, and their explicit expressions are shown in~Fig.~\ref{fig:LNC_example}(a). The LNC problem has the following constraints:
\begin{enumerate}[label=\roman*)]
    \item \emph{Full-rank constraints: }
    \begin{equation}\nonumber
        {\rm rank}\left( \begin{bmatrix}
            {\bf f}_{(t_1,u_1)} & {\bf f}_{(t_3,u_1)} & {\bf f}_{(t_4,u_1)}
        \end{bmatrix} \right) = 3
    \end{equation}
    \begin{equation}\nonumber
        {\rm rank}\left( \begin{bmatrix}
            {\bf f}_{(t_2,u_2)} & {\bf f}_{(t_3,u_2)} & {\bf f}_{(t_4,u_2)}
        \end{bmatrix} \right) = 3
    \end{equation}
    where
    \begin{equation*}
        {\bf f}_{(t_1,u_1)} = \begin{bmatrix}
            x_1 x_{13}\\x_2 x_{13}\\x_3 x_{13}
        \end{bmatrix},\quad
        {\bf f}_{(t_3,u_1)} = \begin{bmatrix}
            x_1 x_{14} x_{17}+ x_{10} x_{15} x_{19} + x_{18}x_{7}\\
            x_2 x_{14} x_{17} + x_{11} x_{15} x_{19} + x_{18} x_8\\
            x_3 x_{14} x_{17} + x_{12} x_{15} x_{19} + x_{18} x_9
        \end{bmatrix},
    \end{equation*}
    \begin{equation*}
        {\bf f}_{(t_4,u_1)} =\begin{bmatrix}
            x_{26} x_4 + x_{27}(x_1 x_{14} x_{20} + x_{10} x_{15} x_{22} + x_{21} x_7)\\
            x_{26} x_5 + x_{27} (x_{11} x_{15} x_{22} + x_{14} x_2 x_{20} + x_{21} x_8)\\
            x_{26} x_6 + x_{27} (x_{12} x_{15} x_{22} + x_{14} x_{20} x_3 + x_{21} x_9)
        \end{bmatrix},
    \end{equation*}
    \begin{equation*}
        {\bf f}_{(t_2,u_2)} = \begin{bmatrix}
           x_{10} x_{16}\\
           x_{11} x_{16}\\
           x_{12} x_{16}
        \end{bmatrix},\quad
        {\bf f}_{(t_3,u_2)} = \begin{bmatrix}
            x_1 x_{14} x_{23} + x_{10} x_{15} x_{25} + x_{24} x_7\\
            x_{11} x_{15} x_{25} + x_{14} x_2 x_{23} + x_{24} x_{8}\\
            x_{12} x_{15} x_{25} + x_{14} x_{23} x_3 + x_{24} x_9
        \end{bmatrix},
    \end{equation*}
    \begin{equation*}
        {\bf f}_{(t_4,u_2)} =\begin{bmatrix}
            x_{28} x_4 + x_{29}(x_1 x_{14} x_{20} + x_{10} x_{15} x_{22} + x_{21} x_7)\\
            x_{28} x_{5} + x_{29}(x_{11} x_{15} x_{22} + x_{14} x_2 x_{20} + x_{21} x_8)\\
            x_{28} x_6 + x_{29}(x_{12} x_{15} x_{22} + x_{14} x_{20} x_3 + x_{21} x_9)
        \end{bmatrix}.
    \end{equation*}
    \item \emph{Non-rank constraints:}
    \begin{equation}\nonumber
    \begin{aligned}
        \text{(Routing:)~}& x_1x_2 = 0,x_1x_3 = 0, x_2x_3 = 0,x_1+x_2+x_3+1 = 0,\\
        & x_4x_5 = 0,x_4x_6 = 0, x_5x_6 = 0,x_4+x_5+x_6+1 = 0,\\
        & x_7x_8 = 0,x_7x_9 = 0, x_8x_9 = 0,x_7+x_8+x_9+1 = 0,\\
        & x_{10}x_{11} = 0,x_{10}x_{12} = 0, x_{11}x_{12} = 0,x_{10}+x_{11}+x_{12}+1 = 0,\\
        & x_{26} x_{27} = 0, x_{26} + x_{27} + 1 = 0,\\
        & x_{28} x_{29} = 0, x_{28} + x_{29} + 1 = 0,\\
        \text{(Broadcast:)~} & x_{17} + x_{20} = 0,x_{20}+x_{23} = 0,\\
        &x_{18}+x_{21} = 0,x_{21}+x_{24} = 0,\\
        &x_{19}+x_{22} = 0,x_{22}+x_{25} = 0.\\      
    \end{aligned}
    \end{equation}
\end{enumerate}

Using the BCSFR algorithm, the feasible set of this problem is explicitly characterized by $33$ characteristic sets $\T_1,\ldots,\T_{33}$, which are provided in Table~\ref{table:lnc_example_results} in the appendices. We summarize the results of the BCSFR algorithm for this example in Table~\ref{table:LNC_illu_example}.

\begin{table}[h]
    \centering
    \caption{Results of the BCSFR Algorithm for the LNC Example in Sec.~\ref{subsec:example_lnc}}
    \label{table:LNC_illu_example}
	\begin{threeparttable}
		\begin{tabular}{>{\centering\arraybackslash}p{1.5cm}>{\centering\arraybackslash}p{3cm}>{\centering\arraybackslash}p{3cm}>{\centering\arraybackslash}p{3cm}}
			\toprule
			Runtime\tnote{\S} & Number of characteristic sets& Number of feasible solutions ($|\mathcal{S}|$) & Search space size (exhaustive search)\tnote{\dag}\\
			\midrule
			0.014s & 33 & 156 & 41,472\\
			\bottomrule
		\end{tabular}
		\begin{tablenotes}
			\footnotesize			
			\item[\S] The runtime of the algorithm was measured on a laptop equipped with an Intel Core Ultra~9 285H processor.
            \item[\dag] The search space size is computed via (\ref{eq:exhaustive_sparce}).
		\end{tablenotes}
	\end{threeparttable}
\end{table}

\section{Application to Storage Codes}
\label{sec-appl-storage}
Linear coding can be used for distributed storage systems (e.g., Windows Azure~\cite{huang2012erasure}, Facebook Analytics Hadoop cluster~\cite{sathiamoorthy2013xoring}), to improve the system efficiency (e.g., reduce storage overhead, increase data reliability). In distributed storage systems, a file that contains large-scale data is stored in a distributed manner in many nodes and the user can recover the file by accessing a subset of the nodes. A main difficulty in such systems is that nodes may fail periodically. Once nodes fail, new redundancy needs to be generated and stored into new nodes, which is known as \emph{repair}. A major repair cost is the number of nodes that participate in the repair process.

\subsection{Fundamentals of Locally Repairable Codes}
Our mathematical model of the storage code is based on locally repairable codes~(LRCs)~\cite{locality2012,lrc2014}. However, we will introduce the LRCs under consideration from a self-contained, network coding-based perspective, which is not identical to the LRCs introduced in~\cite{locality2012,lrc2014}. We first introduce definitions of an LRC, and at the end of this subsection, we provide an example to facilitate the understanding of these definitions.

Let $\ell,\eta$ be positive integers such that $\eta \le \ell$, and let $\mathcal{A} \subset \{1,2,\ldots,\ell\}$. Construct a layered DAG $\mathcal{G}_{\mathcal{A}} = (\mathcal{V}_{\mathcal{A}},\mathcal{E}_{\mathcal{A}})$ as follows:
\begin{enumerate}
    \item The node set is constructed by
    \begin{equation*}
        \mathcal{V}_\mathcal{A} = \{s\} \cup \mathcal{V}_1 \cup \mathcal{V}_2 \cup \mathcal{V}_3,
    \end{equation*}
    where $s$ is a node,
    \begin{equation*}
        {\mathcal{V}_1} = \{v_i:i \in \{1,\ldots,\ell\}\},
    \end{equation*}
    \begin{equation*}
        \mathcal{V}_2 = \{t_i: i \in \{1,\ldots,\ell\}\setminus\mathcal{A}\},
    \end{equation*}
    and
    \begin{equation*}
        \mathcal{V}_3 = \{u_1,\ldots,u_{\binom{\ell}{\eta}}\}.
    \end{equation*}
    \item The edge set is constructed by
    \begin{equation*}
    \begin{aligned}
        \mathcal E_{\mathcal A} = &
\{(s,v_i): i\in\{1,\ldots,\ell\}\} \\
&\cup \{(v_i,t_j): i\in\mathcal A,\ j\in\{1,\ldots,\ell\}\setminus\mathcal A\}\\
&\cup \{(x,u_k): k\in\{1,\ldots,\binom{\ell}{\eta}\},\ x\in \mathcal{C}_k\},
    \end{aligned}
    \end{equation*}
    where $\mathcal{C}_{\mathcal{A}} = \{\mathcal{C}_1,\ldots,\mathcal{C}_{\binom{\ell}{\eta}}\}$ is the collection of all $\eta$-subset of $\{v_i: i\in\mathcal A\} \cup \mathcal{V}_2$.
\end{enumerate}

Let $\omega$ be a positive integer such that $\omega \le \eta$. Following the definition of a network in Sec.~\ref{sec-appl-nc}, we consider a $(\mathcal{G}_{\mathcal{A}},s,\mathcal{U})$ network, where $\mathcal{G}_\mathcal{A}$ is a layered DAG defined above, $s$ is the source node that generates $\omega$ source symbols over $\mathbb{F}_2$, and $\mathcal{U} = \mathcal{V}_3$.

Following Sec.~\ref{sec-appl-nc}, we construct an augmented graph $\tilde{\mathcal G}_\mathcal{A}$ by adding $\omega$ imaginary edges entering node $s$. Then, we can define an LNC on the augmented graph $\tilde{\mathcal G}_\mathcal{A}$.

\begin{define}\label{def:LRC}
    An LNC for $(\mathcal{G}_{\mathcal{A}},s,\mathcal{U})$ is called an LRC if 
    \begin{enumerate}
        \item All nodes except for the source node $s$ are broadcast nodes (see (\ref{eq:broadcast_constraint_LNC})).
        \item The LNC is admissible.
    \end{enumerate}
\end{define}

We next provide an example to demonstrate the practical significance of the mathematical model introduced above.
\begin{example}\label{example:LRC}
    Let $\omega = 2$, $\eta = 2$, and $\ell = 3$. Let a \emph{file} be equally divided into $\omega = 2$ packets. For simplicity, we assume each packet consists of only one symbol $b_i$ over $\mathbb{F}_2$.\footnote{This is the same as the assumption in Sec.~\ref{sec-appl-nc}. This assumption does not affect the mathematical analysis. In practice, each packet often has a large number of symbols and can be represented as a column vector.} We store $b_1$, $b_2$, and $b_1 + b_2$ in $\ell = 3$ (storage) nodes $v_1$, $v_2$, and $v_3$ in a distributed manner. Assume that node $v_1$ fails at some time instant. We will add a new node $t_1$ to \emph{repair} the storage system, which can read data from the surviving nodes $v_2$ and $v_3$. We can store $b_2 + (b_1 + b_2) = b_1$ in node $t_1$. For the repaired storage system with nodes $t_1$, $v_2$, and $v_3$, one can access $(b_1,b_2)$ by any $\eta = 2$ nodes. Such a node repair process is shown in Fig.~\ref{fig:LRC_example1}.

    The above node repair process can be equivalently described by an LNC for $(\mathcal{G}_{\mathcal{A}},s,\mathcal{U})$. We show the layered DAG $\mathcal{G}_{\mathcal{A}}$ in Fig.~\ref{fig:LRC_example2}. In Definition~\ref{def:LRC}, an LRC is defined in the perspective of an LNC. We show the global encoding vector for every edge to define the LNC that is equivalent to the LRC in Fig.~\ref{fig:LRC_example1}:
    \begin{equation*}
        \begin{aligned}
            &{\bf f}_{(s,v_1)} = [1\ \ 0]^\top,\quad {\bf f}_{(s,v_2)} = [0\ \ 1]^\top,\quad {\bf f}_{(s,v_3)} = [1\ \ 1]^\top,\\
            &{\bf f}_{(v_2,t_1)} = {\bf f}_{(v_2,u_1)} = {\bf f}_{(v_2,u_3)} = [0\ \ 1]^\top,\\
            &{\bf f}_{(v_3,t_1)} = {\bf f}_{(v_3,u_2)} = {\bf f}_{(v_3,u_3)} = [1\ \ 1]^\top,\\
            &{\bf f}_{(t_1,u_1)} = {\bf f}_{(t_1,u_2)} = [1\ \ 0]^\top.\\
        \end{aligned}
    \end{equation*}
    Note that the $\binom{\ell}{\eta}$ user nodes guarantee that every $\eta$-subset of $\{v_i: i\in\mathcal A\} \cup \mathcal{V}_2$ can recover the $\omega$ source symbols. This property is of importance for LRCs.
\end{example}

\begin{figure}[!t]
	\centering
	\includegraphics[width=3in]{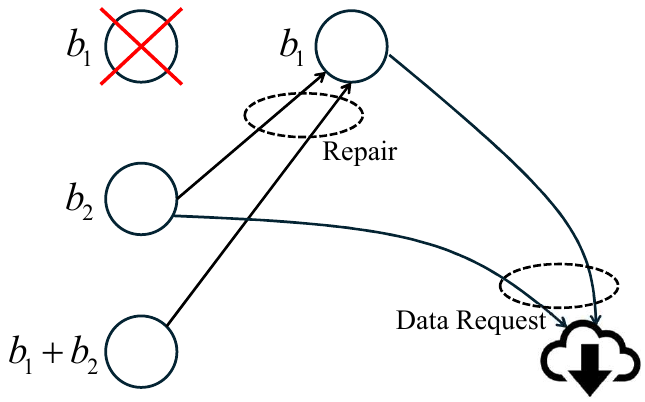}
	\caption{An example of an LRC code for repairing a failed node.}
	\label{fig:LRC_example1}
\end{figure}

\begin{figure}[!t]
	\centering
	\includegraphics[width=3in]{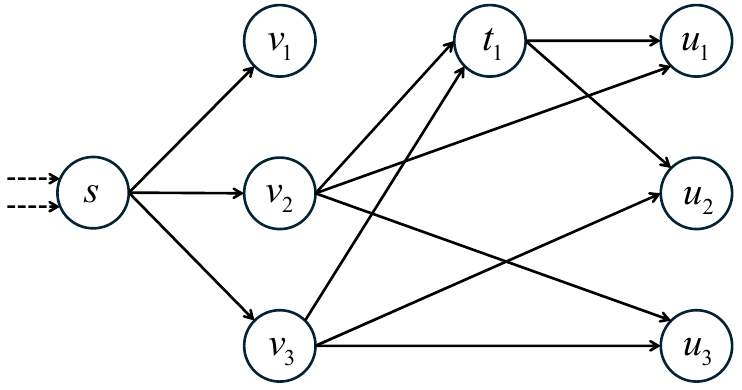}
	\caption{The layered DAG for the LRC in Example~\ref{example:LRC}.}
	\label{fig:LRC_example2}
\end{figure}

\subsection{Problem Formulation}\label{subsec:LRC_formulation}
Consider a layered DAG $\mathcal{G}_\mathcal{A}$. Similar to LNC problems introduced in Sec.~\ref{subsec:appl_nc_formulation}, the local encoding matrices are variables to determine. A slight difference in LRC problems is that several local encoding matrices are constant due to their physical meanings. We classify the local encoding matrices as follows:
\begin{enumerate}
    \item For node $s$: ${\bf G}_s$ is a constant matrix, which determines which linear combinations are stored in the current system.
    \item For $v_i \in \mathcal{V}_1 \setminus \mathcal{A}$: ${\bf G}_{v_i}$ is an empty matrix.
    \item For $v_i \in \mathcal{A}$: ${\bf G}_{v_i}$ is an all-one row vector, which represents that the information stored in $v_i$ can be accessed.
    \item For $t_i \in \mathcal{V}_2$: ${\bf G}_{t_i}$ is an $|\mathcal{A}| \times \binom{\ell-1}{\eta-1}$ matrix which needs to be determined. The broadcast constraint (\ref{eq:broadcast_constraint_LNC}) is imposed to the encoding coefficients of ${\bf G}_{t_i}$.
\end{enumerate}

Following Sec.~\ref{subsec:appl_nc_formulation}, let $\hat{\mathcal{V}} = \mathcal{V}_2$. Then, LRC problems can be formulated using the steps similar to those in Sec.~\ref{subsec:appl_nc_formulation}, see (\ref{eq:rank_constraint_LNC})--(\ref{eq:broadcast_constraint_LNC}).

\begin{remark}
    In our LRC problem, we actually adopt an online design strategy. The layered DAG we define represents a specific state during the operation of the distributed storage system. We design the repair scheme, namely, the local encoding matrices for nodes in $\mathcal{V}_2$, based on the current state, mainly including the linear combinations stored in the surviving nodes. 
    
    Such online designs are reasonable: LRCs are predominantly employed in large-scale archival storage systems, such as Microsoft Azure's Archive Storage, which are designed for ``cold'' data that is infrequently accessed. The repair of a failed node in such environments is treated as a background maintenance task, not a latency-critical operation that directly impacts user experience.
    
    In contrast, the LNC problem in Sec.~\ref{sec-appl-nc} utilizes an offline design strategy. The strict latency constraints of communication usually do not permit online design based on the states of preceding nodes.
\end{remark}

\subsection{An Illustrative Example}\label{subsec:LRC_illu_example}
\begin{figure}[!t]
	\centering
	\includegraphics[width=4in]{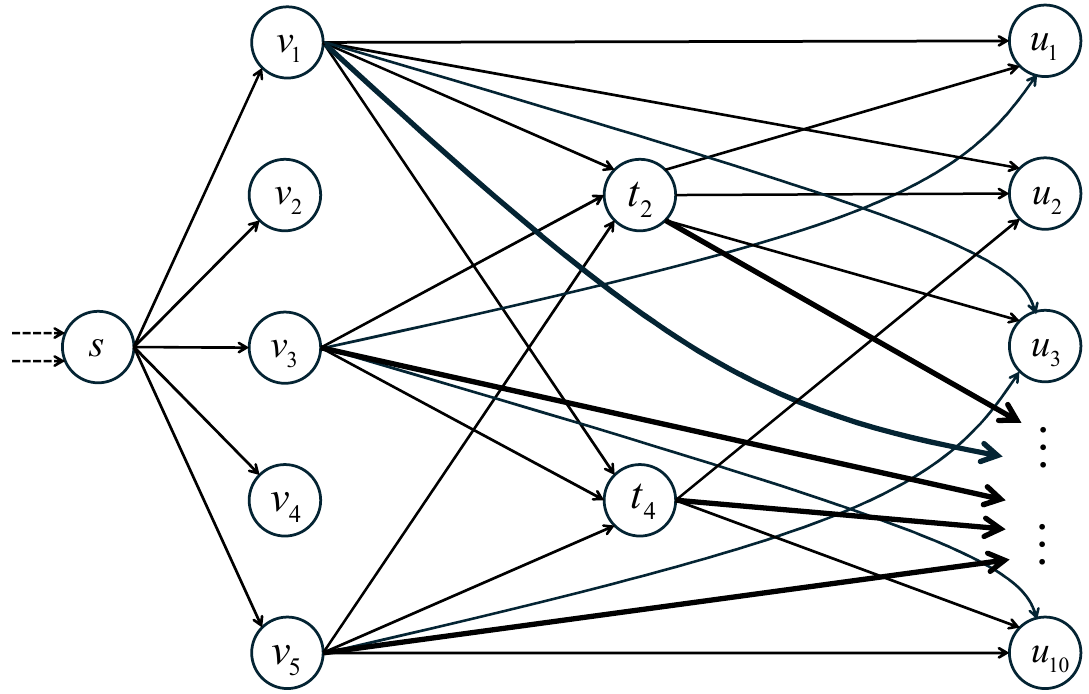}
	\caption{An LRC with $\omega = 2$, $\eta = 3$, and $\ell = 5$. The bold lines denote the edges entering $u_4, \ldots, u_9$ that are not shown in the figure.}
	\label{fig:LRC_example}
\end{figure}
Consider the LRC shown in Fig.~\ref{fig:LRC_example}. Let the local encoding matrix for the source node $s$ be 
\begin{equation*}
    {\bf G}_s = \begin{bmatrix}
        1 & 1 & 1 & 0 & 0\\
        0 & 0 & 1 & 1 & 1
    \end{bmatrix}.
\end{equation*}
Following Sec.~\ref{subsec:LRC_formulation}, ${\bf G}_{v_1}$, ${\bf G}_{v_3}$, and ${\bf G}_{v_5}$ are all-one rows. Let ${\bf G}_{t_2}$ and ${\bf G}_{t_4}$ be
\begin{equation*}
    {\bf G}_{t_2} = \begin{bmatrix}
        x_1 &x_1 &x_1&x_1&x_1&x_1\\
        x_2 &x_2 &x_2&x_2&x_2&x_2\\
        x_3 &x_3 &x_3&x_3&x_3&x_3
    \end{bmatrix},\quad {\bf G}_{t_4} = \begin{bmatrix}
        x_4 &x_4 &x_4&x_4&x_4&x_4\\
        x_5 &x_5 &x_5&x_5&x_5&x_5\\
        x_6 &x_6 &x_6&x_6&x_6&x_6
    \end{bmatrix},
\end{equation*}
where we preprocess the matrices using the broadcast constraints so that every column shares the same variables.

The LRC problem in the form of Problem~\ref{prob:P1} has the constraints:
\begin{enumerate}[label=\roman*)]
    \item \emph{Full-rank constraints: }
    \begin{equation*}
    {\rm rank}\left( [{\bf c}_{i_1}^\top,{\bf c}_{i_2}^\top,{\bf c}_{i_3}^\top] \right) = 2,\quad 1 \le i_1 < i_2 < i_3 \le 5,
\end{equation*}
where 
\begin{equation*}
\begin{aligned}
    &{\bf c}_{1} = (1,0),~{\bf c}_{2} = (x_1+x_2,x_2+x_3),~{\bf c}_{3} = (1,1),\\
    &{\bf c}_{4} = (x_4+x_5,x_5+x_6),~{\bf c}_{5} = (0,1).
    \end{aligned}
\end{equation*}
    \item \emph{Non-rank constraints:}
    \begin{equation*}
        \emptyset.
    \end{equation*}
\end{enumerate}

Using the BCSFR algorithm, the feasible set of this LRC problem is explicitly characterized by 3 characteristic sets $\T_1,\ldots,\T_{3}$, which are provided as follows:
\begin{align*}
    \T_1:\quad & x_2+x_1+1,\,x_3+x_1+1,\,x_6+x_5+1.\\
    \T_2:\quad & x_2+x_1,\,x_3+x_1+1,\,x_5+x_4+1.\\
    \T_3:\quad & x_2+x_1+1,\,x_3+x_1,\,x_6+x_4+1.
\end{align*}
The results of the BCSFR algorithm are summarized in Table~\ref{table:LRC_illu_example}.

\begin{table}[h]
    \centering
    \caption{Results of the BCSFR Algorithm for the LRC Example in Sec.~\ref{subsec:LRC_illu_example}}
    \label{table:LRC_illu_example}
	\begin{threeparttable}
		\begin{tabular}{>{\centering\arraybackslash}p{1.5cm}>{\centering\arraybackslash}p{3cm}>{\centering\arraybackslash}p{3cm}>{\centering\arraybackslash}p{3cm}}
			\toprule
			Runtime\tnote{\S} & Number of characteristic sets& Number of feasible solutions ($|\mathcal{S}|$) & Search space size (exhaustive search)\tnote{\dag}\\
			\midrule
			0.001s & 3 & 24 & 64\\
			\bottomrule
		\end{tabular}
		\begin{tablenotes}
			\footnotesize			
			\item[\S] The runtime of the algorithm was measured on a laptop equipped with an Intel Core Ultra~9 285H processor.
            \item[\dag] The search space size is computed via (\ref{eq:exhaustive_sparce}).
		\end{tablenotes}
	\end{threeparttable}
\end{table}

\section{Experimental Results}
\label{sec-experiment}
In this section, we conduct two experiments to validate the efficacy of the proposed BCSFR/Inc-BCSFR algorithm in designing LNCs and LRCs. In both experiments, we adopt the optimization objective function
\begin{equation}\label{eq:obj_function}
    f({\bf x}) = \sum_{i=1}^n x_i,
\end{equation}
where the summation is taken over the reals. Since this objective function exhibits a positive correlation with the \emph{complexity} of LNCs and LRCs, our objective is to determine the LNC and the LRC that yield the lowest complexity under specific constraints. A comprehensive explanation of the physical meaning of this objective function within the contexts of LNCs and LRCs is provided in the following two subsections.

To the best of our knowledge, there currently exist no other general solvers capable of solving Problem~\ref{prob:P1} with~(\ref{eq:obj_function}) as the objective function, other than an exhaustive search. In our experiments, we consider the cases where $n > 45$, which renders the exhaustive search computationally infeasible on a standard personal computer. 


\subsection{Low-Complexity LNC}\label{subsec:experiment_LNC}
For LNCs, the objective function~(\ref{eq:obj_function}) is designed to quantify the cumulative frequency of memory access events. As mentioned in Sec.~\ref{sec-appl-nc}, the linear combination of symbols in our mathematical model corresponds directly to the linear combination of data packets of large size in practical implementations. Given that large data packets are typically stored in dedicated memory devices, memory access and data fetching operations often dominate the energy consumption and processing latency in practical implementations~\cite{4262451,6757323}. Therefore, our objective function serves to minimize this dominant operational cost. Additionally, we note that this objective function is positively correlated with the number of XOR operations required for the linear combinations, which is also a useful measure of the encoding complexity.

\begin{figure}[!t]
	\centering
	\includegraphics[width=4in]{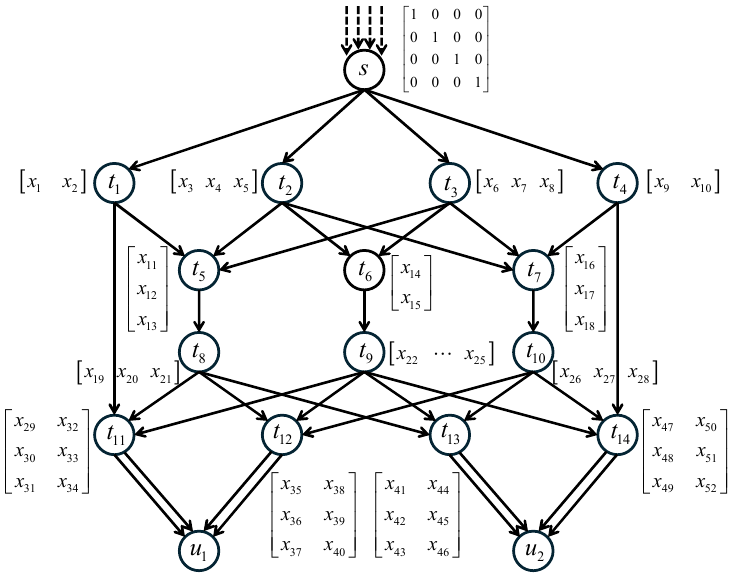}
	\caption{A network $(\mathcal{G},s,\mathcal{U})$ with $\omega = 4$, $\mathcal{V}_{\rm g} = \{t_i:i \in \{1,\ldots,14\} \setminus \{6\}\}$, $\mathcal{V}_{\rm r} = \{t_6\}$, and $\mathcal{V}_{\rm b} = \emptyset$. The local encoding matrix for node $s$ is an identity matrix, and the local encoding matrix for node $t_i$ is represented by $\bf x$.}
	\label{fig:LNC_experiment}
\end{figure}

In our experiment, we consider a network as shown in Fig.~\ref{fig:LNC_experiment}, where $\omega = 4$, $\mathcal{V}_{\rm g} = \{t_i:i \in \{1,\ldots,14\} \setminus \{6\}\}$, $\mathcal{V}_{\rm r} = \{t_6\}$, and $\mathcal{V}_{\rm b} = \emptyset$. The source node $s$ generates $\omega = 4$ packets and transmits them to the user nodes $u_1$ and $u_2$. The local encoding matrices are represented by the binary variables ${\bf x} = (x_1,\ldots,x_{52})$, and their explicit expressions are shown near the nodes in Fig.~\ref{fig:LNC_experiment}. Note that ${\bf G}_s$ is fixed as an identity matrix to reduce the number of variables in this problem.

Using the BCSFR algorithm, we obtain the results in Table~\ref{table:LNC_experiment}. Our algorithm efficiently characterizes the feasible set $\mathcal S$ by 55,910 characteristic sets in 3.67 seconds. In this experiment, we observe that the size of the feasible set is much smaller than the size of the search space, which is reduced by a factor of about $1.885 \times 10^{6}$. Thus, the optimization problem can be solved even by enumerating all points in the feasible set.

\begin{table}[ht]
    \centering
    \caption{Results of the BCSFR Algorithm for the LNC Problem in Sec.~\ref{subsec:experiment_LNC}}
    \label{table:LNC_experiment}
	\begin{threeparttable}
		\begin{tabular}{>{\centering\arraybackslash}p{1.5cm}>{\centering\arraybackslash}p{3cm}>{\centering\arraybackslash}p{3cm}>{\centering\arraybackslash}p{3cm}}
			\toprule
			Runtime\tnote{\S} & Number of characteristic sets& Number of feasible solutions ($|\mathcal{S}|$) & Search space size (exhaustive search)\tnote{\dag}\\
			\midrule
			3.67s & 55,910 & 1,194,393,600 & $2^{51}$\\
			\bottomrule
		\end{tabular}
		\begin{tablenotes}
			\footnotesize			
			\item[\S] The runtime of the algorithm was measured on a laptop equipped with an Intel Core Ultra~9 285H processor.
            \item[\dag] The search space size is computed via (\ref{eq:exhaustive_sparce}).
		\end{tablenotes}
	\end{threeparttable}
\end{table}

With the obtained characteristic sets, we formulate Problem~\ref{prob:P2} with the objective function~(\ref{eq:obj_function}). Using the MaxSAT solver~\cite{Ignatiev2019RC2AE}, Problem~\ref{prob:P2} can be solved efficiently. The number of optimal solutions is 128. We present one of the optimal solutions below, where the listed $x_i$ in \eqref{listx-1} are equal to 1, and the unlisted $x_i$ are equal to~0.
\begin{equation}\label{listx-1}
\begin{array}{ll}
&x_1,\,
x_2,\,
x_4,\,
x_6,\,
x_8,\,
x_9,\,
x_{10},\,
x_{11},\,
x_{12},\,
x_{13},\,
x_{14},\,
x_{17},\,
x_{18},\,
x_{20},\,\\
&x_{21},\,
x_{22},\,
x_{25},\,
x_{26},\,
x_{27},\,
x_{31},\,
x_{32},\,
x_{37},\,
x_{38},\,
x_{43},\,
x_{44},\,
x_{49},\,
x_{50}.
\end{array}
\end{equation}

\subsection{Low-Complexity LRC}\label{subsec:experiment_LRC}
For LRCs, the objective function~(\ref{eq:obj_function}) measures the number of node accesses required for repairing the failed nodes. Note that if the node $i$ participates in the repair of both nodes $j$ and $k$, its contribution to the number of nodes accesses is 2. Minimizing the number of node accesses is an optimization objective of practical significance in large-scale distributed storage systems~\cite{locality2012,lrc2014}.

Let $\omega=18$, $\eta = 23$, $\ell = 25$, $\mathcal{V}_1 = \{v_i: i \in \{1,\ldots,25\}\}$, $\mathcal{A} = \{3,\ldots,25\}$, $\mathcal{V}_2 = \{t_1,t_2\}$, and $\mathcal{V}_3 = \{u_i:i \in \{1,\ldots,\binom{25}{23}\}\}$. We can construct a network $(\mathcal{G}_{\mathcal{A}},s,\mathcal{U})$ following Sec.~\ref{sec-appl-storage}. 
Let ${\bf G}_{t_1}$ and ${\bf G}_{t_2}$ be represented by ${\bf x} = (x_1,\ldots,x_{46})$ as
\begin{equation*}
    {\bf G}_{t_1} = \begin{bmatrix}
        x_1 & \cdots & x_1\\
        \vdots & \cdots & \vdots\\
        x_{23} & \cdots & x_{23}
    \end{bmatrix},\quad {\bf G}_{t_2} = \begin{bmatrix}
        x_{24} & \cdots & x_{24}\\
        \vdots & \cdots & \vdots\\
        x_{46} & \cdots & x_{46}
    \end{bmatrix}.
\end{equation*}
Let the local encoding matrix ${\bf G}_s$ for the source node $s$ be
\begin{equation*}
    {\bf G}_s = \begin{bmatrix}
        0 \,\, 1 \,\, 1 \,\, 0 \,\, 0 \,\, 1 \,\, 1 \,\, 0 \,\, 0 \,\, 0 \,\, 1 \,\, 0 \,\, 0 \,\, 1 \,\, 1 \,\, 0 \,\, 0 \,\, 1 \,\, 0 \,\, 1 \,\, 0 \,\, 0 \,\, 1 \,\, 0 \,\, 1\\
1 \,\, 1 \,\, 1 \,\, 0 \,\, 0 \,\, 1 \,\, 0 \,\, 0 \,\, 0 \,\, 0 \,\, 0 \,\, 0 \,\, 1 \,\, 0 \,\, 1 \,\, 1 \,\, 0 \,\, 0 \,\, 1 \,\, 0 \,\, 1 \,\, 1 \,\, 1 \,\, 0 \,\, 1\\
0 \,\, 0 \,\, 1 \,\, 0 \,\, 0 \,\, 0 \,\, 1 \,\, 1 \,\, 0 \,\, 0 \,\, 0 \,\, 0 \,\, 1 \,\, 1 \,\, 1 \,\, 1 \,\, 1 \,\, 0 \,\, 1 \,\, 1 \,\, 1 \,\, 0 \,\, 0 \,\, 0 \,\, 0\\
1 \,\, 0 \,\, 1 \,\, 0 \,\, 1 \,\, 1 \,\, 0 \,\, 0 \,\, 0 \,\, 1 \,\, 0 \,\, 0 \,\, 1 \,\, 0 \,\, 0 \,\, 0 \,\, 1 \,\, 0 \,\, 1 \,\, 1 \,\, 1 \,\, 0 \,\, 0 \,\, 1 \,\, 0\\
1 \,\, 0 \,\, 0 \,\, 0 \,\, 1 \,\, 0 \,\, 0 \,\, 1 \,\, 0 \,\, 1 \,\, 0 \,\, 1 \,\, 1 \,\, 0 \,\, 0 \,\, 1 \,\, 1 \,\, 0 \,\, 1 \,\, 1 \,\, 0 \,\, 0 \,\, 0 \,\, 1 \,\, 0\\
0 \,\, 0 \,\, 1 \,\, 0 \,\, 1 \,\, 0 \,\, 1 \,\, 1 \,\, 0 \,\, 0 \,\, 1 \,\, 1 \,\, 0 \,\, 0 \,\, 0 \,\, 1 \,\, 0 \,\, 0 \,\, 1 \,\, 0 \,\, 0 \,\, 1 \,\, 0 \,\, 1 \,\, 1\\
0 \,\, 1 \,\, 1 \,\, 0 \,\, 1 \,\, 1 \,\, 1 \,\, 0 \,\, 0 \,\, 0 \,\, 0 \,\, 0 \,\, 0 \,\, 1 \,\, 1 \,\, 1 \,\, 0 \,\, 0 \,\, 0 \,\, 1 \,\, 0 \,\, 0 \,\, 0 \,\, 1 \,\, 1\\
1 \,\, 0 \,\, 1 \,\, 1 \,\, 1 \,\, 1 \,\, 1 \,\, 0 \,\, 1 \,\, 1 \,\, 1 \,\, 1 \,\, 1 \,\, 0 \,\, 0 \,\, 0 \,\, 0 \,\, 1 \,\, 0 \,\, 0 \,\, 1 \,\, 0 \,\, 0 \,\, 1 \,\, 0\\
1 \,\, 0 \,\, 1 \,\, 0 \,\, 0 \,\, 0 \,\, 0 \,\, 1 \,\, 0 \,\, 1 \,\, 0 \,\, 1 \,\, 1 \,\, 1 \,\, 1 \,\, 1 \,\, 0 \,\, 0 \,\, 1 \,\, 0 \,\, 0 \,\, 1 \,\, 0 \,\, 0 \,\, 1\\
0 \,\, 0 \,\, 0 \,\, 1 \,\, 1 \,\, 0 \,\, 0 \,\, 0 \,\, 1 \,\, 1 \,\, 1 \,\, 1 \,\, 1 \,\, 0 \,\, 1 \,\, 1 \,\, 1 \,\, 0 \,\, 1 \,\, 1 \,\, 1 \,\, 1 \,\, 0 \,\, 0 \,\, 1\\
0 \,\, 1 \,\, 1 \,\, 0 \,\, 1 \,\, 0 \,\, 0 \,\, 0 \,\, 0 \,\, 0 \,\, 0 \,\, 1 \,\, 1 \,\, 1 \,\, 1 \,\, 0 \,\, 0 \,\, 1 \,\, 1 \,\, 0 \,\, 1 \,\, 0 \,\, 1 \,\, 0 \,\, 1\\
1 \,\, 1 \,\, 0 \,\, 0 \,\, 1 \,\, 1 \,\, 0 \,\, 1 \,\, 0 \,\, 1 \,\, 1 \,\, 0 \,\, 0 \,\, 0 \,\, 0 \,\, 1 \,\, 1 \,\, 1 \,\, 0 \,\, 0 \,\, 1 \,\, 1 \,\, 0 \,\, 1 \,\, 0\\
0 \,\, 0 \,\, 1 \,\, 0 \,\, 1 \,\, 0 \,\, 0 \,\, 0 \,\, 1 \,\, 0 \,\, 1 \,\, 0 \,\, 1 \,\, 0 \,\, 0 \,\, 1 \,\, 1 \,\, 1 \,\, 1 \,\, 1 \,\, 0 \,\, 1 \,\, 1 \,\, 1 \,\, 1\\
1 \,\, 0 \,\, 1 \,\, 0 \,\, 0 \,\, 0 \,\, 1 \,\, 1 \,\, 1 \,\, 0 \,\, 1 \,\, 1 \,\, 1 \,\, 0 \,\, 0 \,\, 0 \,\, 0 \,\, 1 \,\, 1 \,\, 1 \,\, 1 \,\, 0 \,\, 0 \,\, 0 \,\, 1\\
1 \,\, 0 \,\, 1 \,\, 0 \,\, 1 \,\, 1 \,\, 1 \,\, 1 \,\, 0 \,\, 0 \,\, 0 \,\, 0 \,\, 0 \,\, 1 \,\, 1 \,\, 1 \,\, 1 \,\, 1 \,\, 1 \,\, 1 \,\, 1 \,\, 1 \,\, 0 \,\, 0 \,\, 1\\
0 \,\, 0 \,\, 0 \,\, 0 \,\, 1 \,\, 1 \,\, 0 \,\, 0 \,\, 1 \,\, 1 \,\, 0 \,\, 0 \,\, 0 \,\, 1 \,\, 0 \,\, 1 \,\, 0 \,\, 0 \,\, 0 \,\, 0 \,\, 0 \,\, 0 \,\, 1 \,\, 0 \,\, 1\\
0 \,\, 0 \,\, 0 \,\, 1 \,\, 1 \,\, 0 \,\, 1 \,\, 0 \,\, 0 \,\, 1 \,\, 0 \,\, 0 \,\, 1 \,\, 1 \,\, 0 \,\, 0 \,\, 0 \,\, 0 \,\, 0 \,\, 0 \,\, 0 \,\, 1 \,\, 0 \,\, 0 \,\, 0\\
1 \,\, 0 \,\, 1 \,\, 0 \,\, 1 \,\, 0 \,\, 0 \,\, 1 \,\, 1 \,\, 1 \,\, 0 \,\, 0 \,\, 0 \,\, 1 \,\, 0 \,\, 1 \,\, 1 \,\, 0 \,\, 1 \,\, 0 \,\, 0 \,\, 0 \,\, 0 \,\, 0 \,\, 1
    \end{bmatrix}.
\end{equation*}

Similar to the example in Sec.~\ref{subsec:LRC_illu_example}, we preprocess the matrices using the broadcast constraints so that every column shares the same variables. Then, we can formulate the LRC problem.

Using the Inc-BCSFR algorithm, we obtain the results in Table~\ref{table:LRC_experiment}. Our algorithm efficiently characterizes the feasible set $\mathcal S$ by 24 characteristic sets in 9.2 seconds. In this experiment, although the size of the feasible set is still very large, it is fully characterized by only 24 characteristic sets of simple forms, and these characteristic sets can be effectively utilized by optimizers. Note that the size of the feasible set depends only on the linear coding problem itself, and is independent of our algorithm.

\begin{table}[ht]
    \centering
    \caption{Results of the Inc-BCSFR Algorithm for the LRC Problem in Sec.~\ref{subsec:experiment_LRC}}
    \label{table:LRC_experiment}
	\begin{threeparttable}
		\begin{tabular}{>{\centering\arraybackslash}p{1.5cm}>{\centering\arraybackslash}p{3cm}>{\centering\arraybackslash}p{3cm}>{\centering\arraybackslash}p{3cm}}
			\toprule
			Runtime\tnote{\S} & Number of characteristic sets& Number of feasible solutions ($|\mathcal{S}|$) & Search space size (exhaustive search)\tnote{\dag}\\
			\midrule
			9.2s & 24 & 11,132,555,231,232 & $2^{46}$\\
			\bottomrule
		\end{tabular}
		\begin{tablenotes}
			\footnotesize			
			\item[\S] The runtime of the algorithm was measured on a laptop equipped with an Intel Core Ultra~9 285H processor.
            \item[\dag] The search space size is computed via (\ref{eq:exhaustive_sparce}).
		\end{tablenotes}
	\end{threeparttable}
\end{table}

With the characteristic sets, we have Problem~\ref{prob:P2} with the objective function~(\ref{eq:obj_function}). Using the MaxSAT solver~\cite{Ignatiev2019RC2AE}, Problem~\ref{prob:P2} can be solved efficiently. The number of optimal solutions is 384. We present one of the optimal solutions below, where the listed $x_i$ in \eqref{listx-2} are equal to 1, and the unlisted $x_i$ are equal to 0.
\begin{equation}\label{listx-2}
x_{13},\,x_{31},\,x_{39},\,x_{41},\,x_{45}.
\end{equation}

\section{Conclusion and Discussion}
\label{sec-conc}
In this paper, we consider linear coding problems in which full-rank constraints are imposed on several symbolic matrices. These full-rank constraints guarantee the basic properties of the codes (e.g., decodability and minimum distance), and make the codes feasible. We develop a CS-based method to derive full-rank equivalence conditions of symbolic matrices over the binary field, and present the BCSFR algorithm to effectively characterize the equivalence conditions as the zeros of a series of characteristic sets. 

In linear coding problems, it is of fundamental interest to find a feasible code under which a given objective is optimal or near-optimal, which leads to a mathematical optimization problem. Solving such an optimization problem is difficult because it is not straightforward to determine the feasible set under the full-rank constraints. The equivalence conditions obtained by the BCSFR algorithm give a good characterization of the feasible set by characteristic sets, and thus simplify the problem solving.

To show the effectiveness of the BCSFR algorithm, we conduct several experiments on two types of well-known linear codes, LNCs and LRCs. To the best of our knowledge, there do not exist effective methods in the field of coding to solve the optimization problems in our experiments. Utilizing the BCSFR algorithm combined with the MaxSAT solver, we can find the optimal LNCs and LRCs.

The proposed CS method holds significant potential for broader application. We highlight several promising directions:
\begin{enumerate}
    \item Broadening the utilization of the feasible set: This study focuses on selecting a single optimal solution. Future work could address problems that require multiple solutions or the whole feasible set.
    \item Adapting the algorithmic objective: While our current CS-based method fully characterizes the feasible set, it can be modified to rapidly find a single feasible solution. This will be significantly faster than a complete characterization and highly useful for the problems concerned only with feasibility.
    \item Addressing other linear coding constraints: This work focuses on full-rank constraints. However, the proposed CS-based method can be extended to handle other linear space-related constraints, such as subspace inclusion. Many critical constraints in linear coding can be written as linear space-related constraints.
    \item Generalizing to arbitrary finite fields: The CS-based method holds potential for extension to general finite fields.
\end{enumerate}

\appendices
\section{An Example of the BCSFR Algorithm}\label{appendix:BCSFR_example}

\begin{example}
Consider $n=6$, $v=3$, and $\P = \{f_1,f_2,f_3,f_4\}$, where
\begin{align*}
    f_1 &= x_1x_7+(x_1+1)x_2x_4x_6x_9,\\
	f_2 &= (x_2+x_5)x_7+x_8+x_3x_4x_6x_9,\\
	f_3 &= x_3x_6x_7+x_5x_8,\\
	f_4 &= x_5+1.	
\end{align*}
Note that $x_1,\ldots,x_6$ are $\x$-variables and $x_7,x_8,x_9$ are $\tilde{\x}$-variables.

\noindent\textbf{Input}. $\P = \{f_1,f_2,f_3,f_4\}$.

\noindent\textbf{Step 1}. Let $\P^* = \{\{f_1,f_2,f_3,f_4\}\}$ and $\T^* = \emptyset$.

\noindent\textbf{Step 2}. Since $\P^* = \{\{f_1,f_2,f_3,f_4\}\} \ne \emptyset$, enter the while-loop in Line~2.

\noindent\textbf{Step 3}. 
Select a branch $\mathcal Q = \{f_1,f_2,f_3,f_4\}$, and let $\P^* = \emptyset$ and $\T = \emptyset$.

\noindent\textbf{Step 4}. Do the while-loop in Line~6, see Steps 4.1.1 -- 4.5.5.

\noindent\textbf{Step 4.1.1}. Since $\mathcal Q\backslash\{0\} = \{f_1,f_2,f_3,f_4\} \ne \emptyset$, enter the while-loop in Line~6.

\noindent\textbf{Step 4.1.2.1} (\emph{Reduction by monic polynomials: Lines~7--11}). Choose $f_4 = x_5 + 1$ and substitute $x_5 = 1$ into $\mathcal{Q}$ and $\T$. Then $\mathcal Q = \{f_1, f_5 = x_3x_4x_6x_9+x_8+(x_2+1)x_7, f_6 = x_8+x_3x_6x_7\}$ and $\T = \{x_5+1\}$.

\noindent\textbf{Step 4.1.2.2} (\emph{Reduction by monic polynomials: Lines~7--11}). Choose $f_6 = x_8+x_3x_6x_7$ and substitute $x_8 = x_3x_6x_7$ into $\mathcal{Q}$ and $\T$. Then $\mathcal Q = \{f_1, f_7 = x_3x_4x_6x_9+(x_3x_6+x_2+1)x_7\}$ and $\T = \{x_5+1, x_8+x_3x_6x_7\}$. Observe that there is no $f \in\mathcal Q$ with $I = 1$.


\noindent\textbf{Step 4.1.3} (\emph{Choosing polynomial: Lines 14--21}). Since ${\rm cls}(f) \ge 7$ for all $f \in \mathcal{Q}$, we goto Line~17 and let $k^* = 9$. For $k^* = 9$, $|f \in \mathcal{Q}:{\rm cls}(f) = k^*| = \{f_1,f_7\}$. So, we choose $f_7$ by the 2nd-CSO.

\noindent\textbf{Step 4.1.4} (\emph{Initial decomposition: Lines 22--32}). For $f_7 = x_3x_4x_6x_9+(x_3x_6+x_2+1)x_7$, where $I = x_3x_4x_6$ and $U = (x_3x_6+x_2+1)x_7$, we goto Line~30. Then we update $\P^* = \{\P_1\}$ and $\mathcal Q = \{f_1, x_3x_4x_6+1\}$, where $\P_1 = \{f_1, x_5+1, x_8+x_3x_6x_7, x_3x_4x_6, (x_3x_6+x_2+1)x_7\}$.

\noindent\textbf{Step 4.1.5} (\emph{Substitution: Line 33}). Substituting $x_9 = (x_3x_6+x_2+1)x_7$ into $\mathcal Q$ and $\T$, we obtain $\mathcal Q = \{f_8 = (x_1+1)x_2x_3x_4x_6x_7+x_1x_7, f_{9} = x_3x_4x_6+1\}, \T = \{x_5+1, x_8+x_3x_6x_7, x_9+(x_3x_6+x_2+1)x_7\}$. 

\noindent\textbf{Step 4.2.1}. Since $\mathcal Q\backslash\{0\} = \{f_{8},f_{9}\} \ne \emptyset$, continue the while-loop in Line~6.

\noindent\textbf{Step 4.2.2} (\emph{Reduction by monic polynomials: Lines~7--11}). $\forall f\in \mathcal Q, I \ne 1$.


\noindent\textbf{Step 4.2.3} (\emph{Choosing polynomial: Lines 14--21}). Since $|f \in \mathcal{Q}:{\rm cls}(f) \le n| > 0$, we goto Line~15 and choose $f_{9} = x_3x_4x_6+1$.

\noindent\textbf{Step 4.2.4} (\emph{Initial decomposition: Lines 22--32}). For $f_{9} = x_3x_4x_6+1$, where $I = x_3x_4$ and $U = 1$, we perform a fast decomposition using Lines 23--26. Then we update $\mathcal Q = \{((x_1+1)x_6+1)x_7\}$ and $\T = \{x_3+1, x_4+1, x_5+1, x_8+x_7, x_9+x_2x_7\}$. 

\noindent\textbf{Step 4.2.5} (\emph{Substitution: Line 33}). Substituting $x_6 = 1$ into $\mathcal Q$ and $\T$, we obtain $\mathcal Q = \{f_{10} = ((x_1+1)x_2+x_1)x_7\}$ and $\T = \{x_3+1, x_4+1, x_5+1, x_6+1, x_8+x_7, x_9+x_2x_7\}$. 

\noindent\textbf{Step 4.3.1}. Since $\mathcal Q\backslash\{0\} = \{f_{10}\} \ne \emptyset$, continue the while-loop in Line~6.

\noindent\textbf{Step 4.3.2} (\emph{Reduction by monic polynomials: Lines~7--11}). $\forall f\in \mathcal Q, I \ne 1$.


\noindent\textbf{Step 4.3.3} (\emph{Choosing polynomial: Lines 14--21}). Since ${\rm cls}(f) \ge 7$ for all $f \in \mathcal{Q}$, we goto Line~17 and let $k^* = 9$. For $k^* = 9$, there exists $x_9+x_2x_7$ in $\T$, then update $k^*= 8$. For $k^* = 8$, there exists $x_8+x_7$ in $\T$, then update $k^* = 7$. For $k^* = 7$, we choose $f_{10} = ((x_1+1)x_2+x_1)x_7 \in \mathcal Q$.

\noindent\textbf{Step 4.3.4} (\emph{Initial decomposition: Lines 22--32}). For $f_{10} = ((x_1+1)x_2+x_1)x_7$, where $I = (x_1+1)x_2+x_1$ and $U = 0$, we goto Line~30. Then we update $\P^* = \{\P_1, \P_2\}$ and $\mathcal Q = \{f_{11} = (x_1+1)x_2+x_1+1\}$, where $\P_2 = \{x_3+1, x_4+1, x_5+1, x_6+1, x_8+x_7, x_9+x_2x_7, (x_1+1)x_2+x_1, 0\}$.

\noindent\textbf{Step 4.3.5} (\emph{Substitution: Line 33}). Substituting $x_7 = 0$ into $\mathcal Q$ and $\T$, we obtain $\mathcal Q = \{f_{11}\}$ and $\T = \{x_3+1, x_4+1, x_5+1, x_6+1, x_7, x_8, x_9\}$. 

\noindent\textbf{Step 4.4.1 -- Step 4.4.5}. Similar to Step 4.1.1 -- Step 4.1.5, we obtain $\P^* = \{\P_1, \P_2, \P_3\}$, $\mathcal Q = \{x_1\}$ and $\T = \{x_2+1,x_3+1, x_4+1, x_5+1, x_6+1, x_7, x_8, x_9\}$, where $\P_3 = \{x_3+1, x_4+1, x_5+1, x_6+1, x_7, x_8, x_9, x_1+1\}$.







\noindent\textbf{Step 4.5.1 -- Step 4.5.5}. Similar to Step 4.1.1 -- Step 4.1.5, we obtain $\P^* = \{\P_1, \P_2, \P_3\}$, $\mathcal Q = \emptyset$, and $\T = \{x_1, x_2+1, x_3+1, x_4+1, x_5+1, x_6+1, x_7, x_8, x_9\}$. 

\noindent\textbf{Step 5} (\emph{Line 36}). Since $\mathcal Q = \emptyset$, exit the while-loop in Line~6. Observing $\T$ is admissible, we update $\T^* = \{\T_1\}$, where $\T_1 = {\rm Trunc}_\x(\T) = \{x_1, x_2+1, x_3+1, x_4+1, x_5+1, x_6+1\}$.




\noindent\textbf{Step 6}. Since $\P^* = \{\P_1, \P_2, \P_3\}\ne \emptyset$, continue the while-loop in Line~2.

\noindent\textbf{Step 7}. Select a branch $\mathcal Q = \P_1 = \{f_1, x_5+1, x_8+x_3x_6x_7, f_{12} = x_3x_4x_6, f_{13} = (x_3x_6+x_2+1)x_7\}$, and let $\P^* = \{\P_2, \P_3\}$ and $\T = \emptyset$.

\noindent\textbf{Step 8}. Do the while-loop in Line~6, see Steps 8.1.1 -- 8.3.3.

\noindent\textbf{Step 8.1.1 -- Step 8.1.5}. Similar to Step 4.1.1 -- Step 4.1.5, we obtain $\P^* = \{\P_2,\P_3,\P_4\}$, $\mathcal Q = \{x_1x_7, (x_2+1)x_7, x_3x_4+1\}$, and $\T = \{x_5+1,x_6, x_8\}$, where $\P_4 = \{f_1, f_{13}, x_5+1, x_8+x_3x_6x_7, x_3x_4, 0\}$.








\noindent\textbf{Step 8.2.1 -- Step 8.2.5}. Similar to Step 4.1.1 -- Step 4.1.5, we obtain $\P^* = \{\P_2,\P_3,\P_4\}$, $\mathcal Q = \{x_1x_7, (x_2+1)x_7\}$, and $\T = \{x_3+1, x_4+1, x_5+1, x_6, x_8\}$.







\noindent\textbf{Step 8.3.1}. Since $\mathcal Q\backslash\{0\} \ne \emptyset$, continue the while-loop in Line~6.

\noindent\textbf{Step 8.3.2} (\emph{Reduction by monic polynomials: Lines~7--11}). $\forall f\in Q, I \ne 1$.


\noindent\textbf{Step 8.3.3} (\emph{Choosing polynomial: Lines 14--21}). Since ${\rm cls}(f) \ge 7$ for all $f \in \mathcal{Q}$, we goto Line~17 and let $k^* = 9$. Since ${\rm cls}(f) \ne 9$, $\forall f\in \mathcal Q \cup \T$, \textbf{this branch is pruned}. Let $\T  = \emptyset$ and exit the while-loop in Line~6.

\noindent\textbf{Step 9}. Similarly, we repeat the above process until $\P^* = \emptyset$. We finally obtain $\T^* = \{\T_1,\T_2\}$, where
\begin{align*}
    \T_1 &= \{x_1,x_2+1,x_3+1,x_4+1,x_5+1,x_6+1\},\\
    \T_2 &= \{x_1+1,x_3+1,x_4+1,x_5+1,x_6+1\}.
\end{align*}

\noindent\textbf{Output}. $\T^* = \{\T_1, \T_2\}$.
\end{example}

\section{Numerical Results for the Illustrative Example in Sec.~\ref{subsec:example_lnc}}
We show the 33 characteristic sets obtained by the BCSFR algorithm in the following table, where \textbf{df} in the first line of the table represents the degree of freedom~(see Definition~\ref{def:CS}).
{
\footnotesize
\rowcolors{2}{gray!10}{white}
\begin{longtable}{>{\centering\arraybackslash}p{0.5cm}>{\centering\arraybackslash}p{0.5cm}>{\centering\arraybackslash}p{14cm}}
    \caption{Feasible Set of the LNC Problem in Sec.~\ref{subsec:example_lnc}} \label{table:lnc_example_results} \\

    \toprule 
    \rowcolor{white}
    \textbf{Index} & \textbf{df} & \textbf{Polynomials in the characteristic set} \\
    \midrule 
    \endfirsthead

    \rowcolor{white}
    \multicolumn{3}{c}{{\bfseries \tablename\ \thetable{} Continued from previous page}} \\
    \toprule
    \rowcolor{white}
    \textbf{Index} & \textbf{df} & \textbf{Polynomials in the characteristic set} \\
    \midrule
    \endhead

    \midrule
    \multicolumn{3}{r}{{Continued on next page}} \\
    \endfoot

    \bottomrule 
    \endlastfoot

    $\T_1$&1&$x_1+1,\, x_2,\, x_3,\, x_4,\, x_5,\, x_6+1,\, x_7,\, x_8+1,\, x_9,\, x_{10},\, x_{11}+1,\, x_{12},\, x_{13}+1,\, x_{14}+1,\, x_{15},\, x_{16}+1,\, x_{17}+1,\, x_{18}+1,\, x_{20}+1,\, x_{21}+1,\, x_{22}+x_{19},\, x_{23}+1,\, x_{24}+1,\, x_{25}+x_{19},\, x_{26}+1,\, x_{27},\, x_{28}+1,\, x_{29}$\\
        $\T_2$&1&$x_1+1,\, x_2,\, x_3,\, x_4,\, x_5,\, x_6+1,\, x_7,\, x_8+1,\, x_9,\, x_{10},\, x_{11}+1,\, x_{12},\, x_{13}+1,\, x_{14}+1,\, x_{15}+1,\, x_{16}+1,\, x_{17}+1,\, x_{19}+x_{18}+1,\, x_{20}+1,\, x_{21}+x_{18},\, x_{22}+x_{18}+1,\, x_{23}+1,\, x_{24}+x_{18},\, x_{25}+x_{18}+1,\, x_{26}+1,\, x_{27},\, x_{28}+1,\, x_{29}$\\
        $\T_3$&0&$x_1+1,\, x_2,\, x_3,\, x_4,\, x_5,\, x_6+1,\, x_7,\, x_8,\, x_9+1,\, x_{10},\, x_{11}+1,\, x_{12},\, x_{13}+1,\, x_{14}+1,\, x_{15}+1,\, x_{16}+1,\, x_{17}+1,\, x_{18},\, x_{19}+1,\, x_{20}+1,\, x_{21},\, x_{22}+1,\, x_{23}+1,\, x_{24},\, x_{25}+1,\, x_{26}+1,\, x_{27},\, x_{28}+1,\, x_{29}$\\
        $\T_4$&0&$x_1+1,\, x_2,\, x_3,\, x_4,\, x_5,\, x_6+1,\, x_7,\, x_8,\, x_9+1,\, x_{10},\, x_{11}+1,\, x_{12},\, x_{13}+1,\, x_{14}+1,\, x_{15}+1,\, x_{16}+1,\, x_{17}+1,\, x_{18}+1,\, x_{19}+1,\, x_{20}+1,\, x_{21}+1,\, x_{22}+1,\, x_{23}+1,\, x_{24}+1,\, x_{25}+1,\, x_{26}+1,\, x_{27},\, x_{28}+1,\, x_{29}$\\
        $\T_5$&1&$x_1+1,\, x_2,\, x_3,\, x_4,\, x_5+1,\, x_6,\, x_7,\, x_8,\, x_9+1,\, x_{10},\, x_{11},\, x_{12}+1,\, x_{13}+1,\, x_{14}+1,\, x_{15},\, x_{16}+1,\, x_{17}+1,\, x_{18}+1,\, x_{20}+1,\, x_{21}+1,\, x_{22}+x_{19},\, x_{23}+1,\, x_{24}+1,\, x_{25}+x_{19},\, x_{26}+1,\, x_{27},\, x_{28}+1,\, x_{29}$\\
        $\T_6$&1&$x_1+1,\, x_2,\, x_3,\, x_4,\, x_5+1,\, x_6,\, x_7,\, x_8,\, x_9+1,\, x_{10},\, x_{11},\, x_{12}+1,\, x_{13}+1,\, x_{14}+1,\, x_{15}+1,\, x_{16}+1,\, x_{17}+1,\, x_{19}+x_{18}+1,\, x_{20}+1,\, x_{21}+x_{18},\, x_{22}+x_{18}+1,\, x_{23}+1,\, x_{24}+x_{18},\, x_{25}+x_{18}+1,\, x_{26}+1,\, x_{27},\, x_{28}+1,\, x_{29}$\\
        $\T_7$&0&$x_1+1,\, x_2,\, x_3,\, x_4,\, x_5+1,\, x_6,\, x_7,\, x_8+1,\, x_9,\, x_{10},\, x_{11},\, x_{12}+1,\, x_{13}+1,\, x_{14}+1,\, x_{15}+1,\, x_{16}+1,\, x_{17}+1,\, x_{18},\, x_{19}+1,\, x_{20}+1,\, x_{21},\, x_{22}+1,\, x_{23}+1,\, x_{24},\, x_{25}+1,\, x_{26}+1,\, x_{27},\, x_{28}+1,\, x_{29}$\\
        $\T_8$&0&$x_1+1,\, x_2,\, x_3,\, x_4,\, x_5+1,\, x_6,\, x_7,\, x_8+1,\, x_9,\, x_{10},\, x_{11},\, x_{12}+1,\, x_{13}+1,\, x_{14}+1,\, x_{15}+1,\, x_{16}+1,\, x_{17}+1,\, x_{18}+1,\, x_{19}+1,\, x_{20}+1,\, x_{21}+1,\, x_{22}+1,\, x_{23}+1,\, x_{24}+1,\, x_{25}+1,\, x_{26}+1,\, x_{27},\, x_{28}+1,\, x_{29}$\\
        $\T_9$&1&$x_1,\, x_2,\, x_3+1,\, x_4+1,\, x_5,\, x_6,\, x_7,\, x_8+1,\, x_9,\, x_{10},\, x_{11}+1,\, x_{12},\, x_{13}+1,\, x_{14}+1,\, x_{15},\, x_{16}+1,\, x_{17}+1,\, x_{18}+1,\, x_{20}+1,\, x_{21}+1,\, x_{22}+x_{19},\, x_{23}+1,\, x_{24}+1,\, x_{25}+x_{19},\, x_{26}+1,\, x_{27},\, x_{28}+1,\, x_{29}$\\
        $\T_{10}$&1&$x_1,\, x_2,\, x_3+1,\, x_4+1,\, x_5,\, x_6,\, x_7,\, x_8+1,\, x_9,\, x_{10},\, x_{11}+1,\, x_{12},\, x_{13}+1,\, x_{14}+1,\, x_{15}+1,\, x_{16}+1,\, x_{17}+1,\, x_{19}+x_{18}+1,\, x_{20}+1,\, x_{21}+x_{18},\, x_{22}+x_{18}+1,\, x_{23}+1,\, x_{24}+x_{18},\, x_{25}+x_{18}+1,\, x_{26}+1,\, x_{27},\, x_{28}+1,\, x_{29}$\\
        $\T_{11}$&4&$x_1,\, x_2+1,\, x_3,\, x_4+1,\, x_5,\, x_6,\, x_7,\, x_8,\, x_9+1,\, x_{10},\, x_{11}+1,\, x_{12},\, x_{13}+1,\, x_{16}+1,\, x_{18}+1,\, x_{20}+x_{17},\, x_{21}+1,\, x_{22}+x_{19},\, x_{23}+x_{17},\, x_{24}+1,\, x_{25}+x_{19},\, x_{26}+1,\, x_{27},\, x_{28}+1,\, x_{29}$\\
        $\T_{12}$&2&$x_1,\, x_2,\, x_3+1,\, x_4+1,\, x_5,\, x_6,\, x_7,\, x_8,\, x_9+1,\, x_{10},\, x_{11}+1,\, x_{12},\, x_{13}+1,\, x_{15}+1,\, x_{16}+1,\, x_{18}+x_{14}x_{17}+1,\, x_{19}+1,\, x_{20}+x_{17},\, x_{21}+x_{14}x_{17}+1,\, x_{22}+1,\, x_{23}+x_{17},\, x_{24}+x_{14}x_{17}+1,\, x_{25}+1,\, x_{26}+1,\, x_{27},\, x_{28}+1,\, x_{29}$\\
        $\T_{13}$&1&$x_1,\, x_2+1,\, x_3,\, x_4+1,\, x_5,\, x_6,\, x_7,\, x_8,\, x_9+1,\, x_{10},\, x_{11},\, x_{12}+1,\, x_{13}+1,\, x_{14}+1,\, x_{15},\, x_{16}+1,\, x_{17}+1,\, x_{18}+1,\, x_{20}+1,\, x_{21}+1,\, x_{22}+x_{19},\, x_{23}+1,\, x_{24}+1,\, x_{25}+x_{19},\, x_{26}+1,\, x_{27},\, x_{28}+1,\, x_{29}$\\
        $\T_{14}$&1&$x_1,\, x_2+1,\, x_3,\, x_4+1,\, x_5,\, x_6,\, x_7,\, x_8,\, x_9+1,\, x_{10},\, x_{11},\, x_{12}+1,\, x_{13}+1,\, x_{14}+1,\, x_{15}+1,\, x_{16}+1,\, x_{17}+1,\, x_{19}+x_{18}+1,\, x_{20}+1,\, x_{21}+x_{18},\, x_{22}+x_{18}+1,\, x_{23}+1,\, x_{24}+x_{18},\, x_{25}+x_{18}+1,\, x_{26}+1,\, x_{27},\, x_{28}+1,\, x_{29}$\\
        $\T_{15}$&2&$x_1,\, x_2+1,\, x_3,\, x_4+1,\, x_5,\, x_6,\, x_7,\, x_8+1,\, x_9,\, x_{10},\, x_{11},\, x_{12}+1,\, x_{13}+1,\, x_{15}+1,\, x_{16}+1,\, x_{18}+x_{14}x_{17}+1,\, x_{19}+1,\, x_{20}+x_{17},\, x_{21}+x_{14}x_{17}+1,\, x_{22}+1,\, x_{23}+x_{17},\, x_{24}+x_{14}x_{17}+1,\, x_{25}+1,\, x_{26}+1,\, x_{27},\, x_{28}+1,\, x_{29}$\\
        $\T_{16}$&4&$x_1,\, x_2,\, x_3+1,\, x_4+1,\, x_5,\, x_6,\, x_7,\, x_8+1,\, x_9,\, x_{10},\, x_{11},\, x_{12}+1,\, x_{13}+1,\, x_{16}+1,\, x_{18}+1,\, x_{20}+x_{17},\, x_{21}+1,\, x_{22}+x_{19},\, x_{23}+x_{17},\, x_{24}+1,\, x_{25}+x_{19},\, x_{26}+1,\, x_{27},\, x_{28}+1,\, x_{29}$\\
        $\T_{17}$&4&$x_1,\, x_2+1,\, x_3,\, x_4,\, x_5,\, x_6+1,\, x_7+1,\, x_8,\, x_9,\, x_{10},\, x_{11}+1,\, x_{12},\, x_{13}+1,\, x_{16}+1,\, x_{18}+1,\, x_{20}+x_{17},\, x_{21}+1,\, x_{22}+x_{19},\, x_{23}+x_{17},\, x_{24}+1,\, x_{25}+x_{19},\, x_{26}+1,\, x_{27},\, x_{28}+1,\, x_{29}$\\
        $\T_{18}$&4&$x_1,\, x_2,\, x_3+1,\, x_4,\, x_5+1,\, x_6,\, x_7+1,\, x_8,\, x_9,\, x_{10},\, x_{11},\, x_{12}+1,\, x_{13}+1,\, x_{16}+1,\, x_{18}+1,\, x_{20}+x_{17},\, x_{21}+1,\, x_{22}+x_{19},\, x_{23}+x_{17},\, x_{24}+1,\, x_{25}+x_{19},\, x_{26}+1,\, x_{27},\, x_{28}+1,\, x_{29}$\\
        $\T_{19}$&2&$x_1+1,\, x_2,\, x_3,\, x_4,\, x_5,\, x_6+1,\, x_7+1,\, x_8,\, x_9,\, x_{10},\, x_{11}+1,\, x_{12},\, x_{13}+1,\, x_{15}+1,\, x_{16}+1,\, x_{18}+x_{14}x_{17}+1,\, x_{19}+1,\, x_{20}+x_{17},\, x_{21}+x_{14}x_{17}+1,\, x_{22}+1,\, x_{23}+x_{17},\, x_{24}+x_{14}x_{17}+1,\, x_{25}+1,\, x_{26}+1,\, x_{27},\, x_{28}+1,\, x_{29}$\\
        $\T_{20}$&2&$x_1+1,\, x_2,\, x_3,\, x_4,\, x_5+1,\, x_6,\, x_7+1,\, x_8,\, x_9,\, x_{10},\, x_{11},\, x_{12}+1,\, x_{13}+1,\, x_{15}+1,\, x_{16}+1,\, x_{18}+x_{14}x_{17}+1,\, x_{19}+1,\, x_{20}+x_{17},\, x_{21}+x_{14}x_{17}+1,\, x_{22}+1,\, x_{23}+x_{17},\, x_{24}+x_{14}x_{17}+1,\, x_{25}+1,\, x_{26}+1,\, x_{27},\, x_{28}+1,\, x_{29}$\\
        $\T_{21}$&1&$x_1,\, x_2,\, x_3+1,\, x_4+1,\, x_5,\, x_6,\, x_7+1,\, x_8,\, x_9,\, x_{10},\, x_{11}+1,\, x_{12},\, x_{13}+1,\, x_{14}+1,\, x_{15}+1,\, x_{16}+1,\, x_{17}+1,\, x_{19}+1,\, x_{20}+1,\, x_{21}+x_{18},\, x_{22}+1,\, x_{23}+1,\, x_{24}+x_{18},\, x_{25}+1,\, x_{26}+1,\, x_{27},\, x_{28}+1,\, x_{29}$\\
        $\T_{22}$&1&$x_1,\, x_2+1,\, x_3,\, x_4+1,\, x_5,\, x_6,\, x_7+1,\, x_8,\, x_9,\, x_{10},\, x_{11},\, x_{12}+1,\, x_{13}+1,\, x_{14}+1,\, x_{15}+1,\, x_{16}+1,\, x_{17}+1,\, x_{19}+1,\, x_{20}+1,\, x_{21}+x_{18},\, x_{22}+1,\, x_{23}+1,\, x_{24}+x_{18},\, x_{25}+1,\, x_{26}+1,\, x_{27},\, x_{28}+1,\, x_{29}$\\
        $\T_{23}$&2&$x_1,\, x_2+1,\, x_3,\, x_4,\, x_5,\, x_6+1,\, x_7,\, x_8+1,\, x_9,\, x_{10}+1,\, x_{11},\, x_{12},\, x_{13}+1,\, x_{15}+1,\, x_{16}+1,\, x_{18}+x_{14}x_{17}+1,\, x_{19}+1,\, x_{20}+x_{17},\, x_{21}+x_{14}x_{17}+1,\, x_{22}+1,\, x_{23}+x_{17},\, x_{24}+x_{14}x_{17}+1,\, x_{25}+1,\, x_{26}+1,\, x_{27},\, x_{28}+1,\, x_{29}$\\
        $\T_{24}$&0&$x_1,\, x_2+1,\, x_3,\, x_4,\, x_5,\, x_6+1,\, x_7,\, x_8,\, x_9+1,\, x_{10}+1,\, x_{11},\, x_{12},\, x_{13}+1,\, x_{14}+1,\, x_{15}+1,\, x_{16}+1,\, x_{17}+1,\, x_{18},\, x_{19}+1,\, x_{20}+1,\, x_{21},\, x_{22}+1,\, x_{23}+1,\, x_{24},\, x_{25}+1,\, x_{26}+1,\, x_{27},\, x_{28}+1,\, x_{29}$\\
        $\T_{25}$&1&$x_1,\, x_2,\, x_3+1,\, x_4,\, x_5+1,\, x_6,\, x_7,\, x_8+1,\, x_9,\, x_{10}+1,\, x_{11},\, x_{12},\, x_{13}+1,\, x_{14}+1,\, x_{15}+1,\, x_{16}+1,\, x_{17}+1,\, x_{19}+1,\, x_{20}+1,\, x_{21}+x_{18},\, x_{22}+1,\, x_{23}+1,\, x_{24}+x_{18},\, x_{25}+1,\, x_{26}+1,\, x_{27},\, x_{28}+1,\, x_{29}$\\
        $\T_{26}$&2&$x_1,\, x_2,\, x_3+1,\, x_4,\, x_5+1,\, x_6,\, x_7,\, x_8,\, x_9+1,\, x_{10}+1,\, x_{11},\, x_{12},\, x_{13}+1,\, x_{15}+1,\, x_{16}+1,\, x_{18}+x_{14}x_{17}+1,\, x_{19}+1,\, x_{20}+x_{17},\, x_{21}+x_{14}x_{17}+1,\, x_{22}+1,\, x_{23}+x_{17},\, x_{24}+x_{14}x_{17}+1,\, x_{25}+1,\, x_{26}+1,\, x_{27},\, x_{28}+1,\, x_{29}$\\
        $\T_{27}$&0&$x_1,\, x_2+1,\, x_3,\, x_4,\, x_5,\, x_6+1,\, x_7,\, x_8,\, x_9+1,\, x_{10}+1,\, x_{11},\, x_{12},\, x_{13}+1,\, x_{14}+1,\, x_{15}+1,\, x_{16}+1,\, x_{17}+1,\, x_{18}+1,\, x_{19}+1,\, x_{20}+1,\, x_{21}+1,\, x_{22}+1,\, x_{23}+1,\, x_{24}+1,\, x_{25}+1,\, x_{26}+1,\, x_{27},\, x_{28}+1,\, x_{29}$\\
        $\T_{28}$&4&$x_1+1,\, x_2,\, x_3,\, x_4,\, x_5,\, x_6+1,\, x_7,\, x_8+1,\, x_9,\, x_{10}+1,\, x_{11},\, x_{12},\, x_{13}+1,\, x_{16}+1,\, x_{18}+1,\, x_{20}+x_{17},\, x_{21}+1,\, x_{22}+x_{19},\, x_{23}+x_{17},\, x_{24}+1,\, x_{25}+x_{19},\, x_{26}+1,\, x_{27},\, x_{28}+1,\, x_{29}$\\
        $\T_{29}$&4&$x_1+1,\, x_2,\, x_3,\, x_4,\, x_5+1,\, x_6,\, x_7,\, x_8,\, x_9+1,\, x_{10}+1,\, x_{11},\, x_{12},\, x_{13}+1,\, x_{16}+1,\, x_{18}+1,\, x_{20}+x_{17},\, x_{21}+1,\, x_{22}+x_{19},\, x_{23}+x_{17},\, x_{24}+1,\, x_{25}+x_{19},\, x_{26}+1,\, x_{27},\, x_{28}+1,\, x_{29}$\\
        $\T_{30}$&1&$x_1,\, x_2+1,\, x_3,\, x_4,\, x_5,\, x_6+1,\, x_7+1,\, x_8,\, x_9,\, x_{10}+1,\, x_{11},\, x_{12},\, x_{13}+1,\, x_{14}+1,\, x_{15},\, x_{16}+1,\, x_{17}+1,\, x_{18}+1,\, x_{20}+1,\, x_{21}+1,\, x_{22}+x_{19},\, x_{23}+1,\, x_{24}+1,\, x_{25}+x_{19},\, x_{26}+1,\, x_{27},\, x_{28}+1,\, x_{29}$\\
        $\T_{31}$&1&$x_1,\, x_2+1,\, x_3,\, x_4,\, x_5,\, x_6+1,\, x_7+1,\, x_8,\, x_9,\, x_{10}+1,\, x_{11},\, x_{12},\, x_{13}+1,\, x_{14}+1,\, x_{15}+1,\, x_{16}+1,\, x_{17}+1,\, x_{19}+x_{18}+1,\, x_{20}+1,\, x_{21}+x_{18},\, x_{22}+x_{18}+1,\, x_{23}+1,\, x_{24}+x_{18},\, x_{25}+x_{18}+1,\, x_{26}+1,\, x_{27},\, x_{28}+1,\, x_{29}$\\
        $\T_{32}$&1&$x_1,\, x_2,\, x_3+1,\, x_4,\, x_5+1,\, x_6,\, x_7+1,\, x_8,\, x_9,\, x_{10}+1,\, x_{11},\, x_{12},\, x_{13}+1,\, x_{14}+1,\, x_{15},\, x_{16}+1,\, x_{17}+1,\, x_{18}+1,\, x_{20}+1,\, x_{21}+1,\, x_{22}+x_{19},\, x_{23}+1,\, x_{24}+1,\, x_{25}+x_{19},\, x_{26}+1,\, x_{27},\, x_{28}+1,\, x_{29}$\\
        $\T_{33}$&1&$x_1,\, x_2,\, x_3+1,\, x_4,\, x_5+1,\, x_6,\, x_7+1,\, x_8,\, x_9,\, x_{10}+1,\, x_{11},\, x_{12},\, x_{13}+1,\, x_{14}+1,\, x_{15}+1,\, x_{16}+1,\, x_{17}+1,\, x_{19}+x_{18}+1,\, x_{20}+1,\, x_{21}+x_{18},\, x_{22}+x_{18}+1,\, x_{23}+1,\, x_{24}+x_{18},\, x_{25}+x_{18}+1,\, x_{26}+1,\, x_{27},\, x_{28}+1,\, x_{29}$\\
\end{longtable}
}

\bibliographystyle{IEEEtran}
\bibliography{reference}






\end{document}